\documentclass{fac}
\usepackage{proof}
\usepackage{amsmath}
\usepackage{amssymb}
\usepackage{stmaryrd}
\usepackage{xcolor}
\usepackage{comment}
\newcommand\black{\ensuremath{\blacktriangleright}}
\newcommand\white{\ensuremath{\vartriangleright}}

\newif\ifamsfontsloaded
\ifprodtf
  \newcommand\whbl{\white\kern-.1em--\kern-.1em\black}
  \newcommand\blwh{\black\kern-.1em--\kern-.1em\white}
  \newcommand\blbl{\black\kern-.1em--\kern-.1em\black}
  \newcommand\whwh{\white\kern-.1em--\kern-.1em\white}
  \amsfontsloadedtrue
\else
  \checkfont{msam10}
  \iffontfound
    \IfFileExists{amssymb.sty}
      {\usepackage{amssymb}\amsfontsloadedtrue
       \newcommand\whbl{\white\kern-.125em--\kern-.125em\black}%
       \newcommand\blwh{\black\kern-.125em--\kern-.125em\white}%
       \newcommand\blbl{\black\kern-.125em--\kern-.125em\black}%
       \newcommand\whwh{\white\kern-.125em--\kern-.125em\white}}
      {}
  \fi
\fi

\newtheorem{example}{Example}
\newtheorem{theorem}{Theorem}%[section]
\newtheorem{lemma}{Lemma}
\newtheorem{definition}{Definition}
\newcommand{\interp}[1]{\llbracket #1 \rrbracket} 
\newtheorem{corollary}{Corollary}
\newtheorem{proposition}{Proposition}
\newcommand{\remph}[1]{\textcolor{blue}{#1}}

\title[Operational Semantics of Resolution and Productivity in Horn Clause Logic]
      {Operational Semantics of Resolution  and Productivity in Horn Clause Logic}

\author[P. Fu and E. Komendantskaya]
    {Peng Fu$^1$ and Ekaterina Komendantskaya$^1$\\
     $^1$Heriot-Watt University, Edinburgh, UK}

\correspond{Peng Fu, e-mail: peng-fu@uiowa.edu, Ekaterina Komendantskaya, e-mail: ek19@hw.ac.uk}

\pubyear{2000}
\pagerange{\pageref{firstpage}--\pageref{lastpage}}

\begin{document}
\label{firstpage}

\makecorrespond

\maketitle

\begin{abstract}
  This paper presents a study of operational and type-theoretic properties of different resolution strategies in Horn clause logic. % using the framework of abstract reduction systems. 
  We distinguish four different kinds of resolution: resolution by unification (SLD-resolution), resolution by term-matching, the recently introduced structural resolution, and partial (or lazy)
  resolution. We express them all uniformly as abstract reduction systems, which allows us to undertake a thorough comparative analysis of their properties.
  To match this small-step semantics,
  we propose to take Howard's System $\mathbf{H}$ as a type-theoretic semantic counterpart.
  Using System $\mathbf{H}$, we interpret Horn formulas  as types, and
  a derivation for a given formula as the proof term inhabiting the type given by the formula.
  We prove soundness of these abstract reduction systems relative to System $\mathbf{H}$, and we show completeness of SLD-resolution and structural resolution relative to System $\mathbf{H}$.
  %While the completeness property holds for SLD-resolution and structural resolution. 
  We identify conditions under which  structural resolution is operationally equivalent to
   SLD-resolution.
We show correspondence between  term-matching resolution for Horn clause programs without existential variables and term rewriting.


\end{abstract}
\section{Introduction}
\label{intro}

%% Logic programming (LP) is an automated theorem prover and programming
%% language based on first-order Horn clause logic. 
Horn clause logic is a fragment of first-order logic that gives theoretical foundation to logic programming. A set of Horn clauses is called a logic program. 
SLD-resolution is the most common algorithm in logic programming for automatically inferring whether, given a logic program
$\Phi$ and a first-order formula $A$,  $\Phi \vdash \sigma A$ holds for some substitution $\sigma$. 
SLD-resolution is  semi-decidable as not all derivations by 
SLD-resolution  terminate.
Terminating SLD-resolution is quite well understood (see for example the textbook by Lloyd~\cite{Llo88}), and SLD-resolution has been successfully incorporated into a number of logic 
programming language implementations. 
However, nonterminating SLD-resolution is more challenging to handle.
% not as well understood, and it poses a challenge to compute with infinite data structures. 

\begin{example}\label{ex:zstream} 
Consider the following logic program defining the infinite stream of zeros. It consists of one Horn clause:

\begin{center}
  $\kappa_1:\forall y.  Stream(y)  \Rightarrow Stream(Cons(0, y))$
\end{center}

\noindent For query $Stream(x)$, it gives rise to the following SLD-derivation: 

  \begin{center}
    $\Phi \vdash \{{Stream}(x)\} \leadsto_{\kappa_1, [{Cons}(0, y_1)/x]} \{{Stream}(y_1)\} \leadsto_{\kappa_1, [Cons(0, y_2)/y_1, {Cons}(0, Cons(0, y_2))/x]} 
    \{{Stream}(y_2)\} \leadsto_{\kappa_1,  [{Cons}(0, y_3)/y_2, Cons(0, {Cons}(0, y_3))/y_1, {Cons}(0, Cons(0, {Cons}(0, y_3)))/x]} 
    \{{Stream}(y_3)\} \leadsto ... $
  \end{center}

  \noindent At each derivation step, we record the clause that is used to make the resolution step and the computed substitution.
  For this derivation, it is impossible to find a finite substitution $\sigma$ such that $\Phi \vdash \sigma(Stream(x))$.
  Nevertheless, the query $Stream(x)$ is computationally 
meaningful, since it computes the infinite stream of zeros for $x$.

\end{example}

The importance of developing approaches for computing with infinite data structures in logic programming
has been argued by van Emden, Lloyd, \emph{et al} \cite{Llo88} in the 80s and more recently the topic has been revived by Gupta, Simon \emph{et al.}~\cite{Gupta07,Simon07}.
In the classical approach~\cite{Llo88}, a semantic view was taken: if a nonterminating SLD-resolution derivation for $\Phi$ and $A$
accumulates computed substitutions
$\sigma_0, \sigma_1 , \ldots $ in such a way that $ (\ldots (\sigma_1(\sigma_0(A))))$ is an infinite ground formula, then $(\ldots (\sigma_1(\sigma_0(A))))$ is said to be \textit{computable at infinity}.
Computation at infinity is proven to be sound with respect to the greatest Herbrand model, i.e., given a logic program $\Phi$,  if a formula $A$ is computable at infinity with respect to $\Phi$, then $A$
 is also in the greatest Herbrand model of $\Phi$. Importantly for us, the notion of \textit{infinite formula computed at infinity} captures the modern-day notion of producing an infinite data structure. %, in modern terms.
We will use the terminology \emph{global productivity} to describe computation at infinity. For example, the derivation shown in Example~\ref{ex:zstream} is globally productive, as it computes the infinite stream of zeros at infinity.
However, this approach did not result in implementation and in general proving global productivity
 is nontrivial.

An alternative approach has been proposed by Gupta, Simon \emph{et al.}~\cite{Gupta07,Simon07}: subgoals produced in the course of an infinite SLD-derivation can be memorized, and if any two subgoals are unifiable,
then the derivation is said to be closed coinductively.
This approach was implemented as an extension to Prolog and called CoLP (Coinductive Logic Programming). Its applications are limited, as % but it suffers from a main problem:
CoLP does not terminate for the SLD derivation that produces an infinite formula with irrational tree structure, as in this case the derivation does not feature any unifiable subgoals.
CoLP's approach was not intended to capture the notion of  global productivity. % , it is not suitable for  the study of productivity of infinite data computed by resolution.
%it is not very helpful for distinguishing cycling nonterminating derivations from those computing infinite data.
For example, the query $P(x)$ for the clause $P(y) \Rightarrow P(y)$ will exhibit a cycle and will be coinductively proven by CoLP, but it will not compute an infinite formula at infinity. In other words, the derivation for $P(x)$ is not globally productive despite being coinductively provable by CoLP.

In this paper, we introduce yet another approach to the potentially infinite derivations by the SLD-resolution. %  notion of productivity for LP. 
When SLD-resolution produces a finite or infinite ground answer for a variable in the query, we say the query is \textit{locally productive} at that variable (see Definition \ref{local}).
This gives us an alternative notion of productivity for logic programming.
In order to formally define this notion of local productivity, we introduce a lazy version of
resolution (called \emph{partial resolution}). Firstly, we label those variables in the queries for which we want to compute substitutions.
Partial resolution then takes the labels into account when performing the computation, by giving priority to subgoals with labelled variables.
The resolution will stop when all the labels in the queries are eliminated, or in other words, when all required substitutions have been computed.

%% In this paper, we propose an alternative approach to analysing nonterminating derivations and distinguishing those producing infinite data.
%% This approach is based on introducing a more fine-grained analysis of operational semantics of resolution for LP.
Finally,  a third notion of productivity for logic programming, an \emph{observational productivity}, has been recently introduced by Johann, Komendantskaya \emph{et al}~\cite{JKK15,KJS16}. %Unlike any of the three above approaches,
It depends on a new kind of resolution (\emph{structural resolution}). %, rather than the SLD-resolution.
Structural resolution depends crucially on term-matching resolution, obtained by restricting unification used in the SLD-resolution to term matching.
%As we show in this paper, structural resolution is a combination of two variants of the classical SLD-resolution, both have distinct operational properties and play a role in productivity studies.
%The usual unification-based resolution known and used in LP is under the name of SLD-resolution.
%If we replace the unification in SLD-resolution by term-matching, we obtain term-matching resolution. % (also called \textit{LP-TM} in this paper). % restricts the use of unification to term-matching, 
 Term-matching resolution is used in e.g. type class resolution~\cite{Jones97} in functional programming. 
It has different operational properties compared to the SLD-resolution.
For example, taking the program in Example \ref{ex:zstream}, the query $Stream(x)$ can not be reduced
by term-matching
resolution, as it is not possible to match $Stream(Cons(0, y))$ with $Stream(x)$.

Structural resolution combines terminating term matching steps with unification steps.
%It is called \emph{structural resolution} in~\cite{JKK15}. % and LP-Struct throughout this paper. 
For example, consider the following derivation (where $\to$ denotes a term-matching step, and $\hookrightarrow$  applies the substitution obtained by unification to the current query): 
  \begin{center}
    $\Phi \vdash \{{Stream}(x)\} \hookrightarrow_{\kappa_1, [{Cons}(0, y_1)/x]} \{{Stream}({Cons}(0,y_1))\} \to_{\kappa_1}
    \{{Stream}(y_1)\}  \hookrightarrow_{\kappa_1, [{Cons}(0, y_2)/y_1, {Cons}(0, {Cons}(0, y_2))/x ]} \{{Stream}({Cons}(0, y_2))\} \to_{\kappa_1} 
    \{{Stream}(y_2)\} \hookrightarrow_{\kappa_1,  [{Cons}(0, y_3)/y_2, {Cons}(0,{Cons}(0, y_3) )/y_1, {Cons}(0, {Cons}(0, {Cons}(0, y_3)))/x]} \{{Stream}({Cons}(0, y_3))\} \to_{\kappa_1} \{{Stream}(y_3)\} \hookrightarrow ... $
  \end{center}

  \noindent Note that the overall derivation is nonterminating, but all term-matching derivations are finite in the above resolution trace.
  This separation of term-matching and unification allows the formulation of \emph{observational productivity}: given a program $\Phi$ and a query $A$, if a derivation for
  $A$ is infinite, and it features only terminating term-matching resolution steps, then this derivation is called \emph{observationally productive}.
%Under certain condition, observational productivity implies global productivity.
  As discussed by Komendantskaya, Johann \emph{et al.}~\cite{KJ15,KJS16}  observational productivity implies global productivity.
  %, we refer to the work by .%% , and required \emph{groundness assumption} for formulas computed at infinity. 
 
 %% Notably this kind of resolution is defined as a partial version of SLD-resolution, 
  %% thus giving us an operational notion of productivity for SLD-resolution.

  %For local productivity, the infinite object is the
  %answer for some variable, it is observational since the answer is required to be ground. 
  
%\end{enumerate}

  To illustrate these three notions of productivity, we consider three logic programs in
the following table. 

%\knote{do you have to postulate groundness as a condition everywhere?}

\begin{center}
  {
    \begin{tabular}{|p{2cm}|  p{2.5cm}  p{4.5cm} p{4.5cm}|}
      %\hline 
      Name & $\Phi_1$ & $\Phi_2$ & $\Phi_3$   \\ & & & \\
      Program & $P(x) \Rightarrow P(x)$  & $P(x) \Rightarrow P(K(x))$ & $P(x, y) \Rightarrow P(x, G(y))$ 
      \\ & & &\\
      Query &  $P(x)$ & $P(x)$ & $P(x,y)$ \\ & & & \\
      Productivity & None & Global,  Observational, Local at $x$ &  Observational, Local at $y$  \\
%Productivity defined for: & SLD-Resolution & Structural Resolution & SLD-Resolution \\
%      \hline
    \end{tabular}
  }
\end{center}

\begin{itemize}
\item Program $\Phi_1$ is not productive for the query $P(x)$ by any of these three notions of productivity: it does 
not compute an infinite ground formula, it is not terminating by term-matching resolution, and it does not compute a ground answer for $x$. 

%\knote{can we avoid the groundness condition? it was not needed for either global or observational cases...}

\item Program $\Phi_2$ is globally productive for the query $P(x)$ as it computes an infinite formula $P(K(K(...)))$.
We can see that $\Phi_2$ is observationally productive because it is terminating by term-matching resolution.
Also, $\Phi_2$ is locally productive at $x$ for the query $P(x)$ since SLD-resolution computes a ground infinite answer $K(K(...))$ for the variable $x$.

%\knote{I have changed the below, as what it used to say about global productivity was not true.}

\item  Program $\Phi_3$ is not globally productive for the query $P(x,y)$ as SLD-resolution
  computes an infinite formula $P(x, G(G(...)))$ that is not ground.
  %This is because no substitution is computed at all in the course of the derivation.   but
  It is observationally productive, since the second argument for $P$
is  decreasing from right to left by the subterm relation, i.e. $y <_{\mathrm{subterm}} G(y)$ and that ensures termination of term-matching resolution. Note that $\Phi_3$
is not locally productive at $x$ for the query $P(x, y)$, but it is locally productive at $y$ since $G(G(...))$ is an infinite ground answer for $y$. 

\end{itemize}

%Ideally, we would like to be able to study the exact relation between these different notions
%of productivity.
%\knote{The trouble is, -- we do not have a single lemma stating the relatons of these three kinds of productivity! So, I'll try to reformulate this bit...}

In this paper, we establish a framework for a
comparative analysis of different kinds of resolution and different notions of productivity.
Firstly, we use a uniform style of small-step semantics for all these kinds of resolution and
formulate them as abstract reduction systems. We call the resulting abstract
reduction systems \emph{LP-Unif}, \emph{partial LP-Unif}, \emph{LP-TM}, and \emph{LP-Struct},  respectively.
Using this framework, we ask and answer several research questions about operational properties and relations of these reduction systems.
Are LP-Unif and LP-Struct equivalent for terminating derivations, and under what conditions?
Are LP-Unif and LP-Struct equivalent for observationally productive programs?
Since the termination of LP-TM is essential for the observational productivity, are there any suitable 
program transformation methods to ensure LP-TM termination?

%The standard SLD-resolution used in LP is formulated as an abstract reduction system called
%  \textit{LP-Unif}. Term-matching resolition is then given by a reduction system \textit{LP-TM}.
%  Finally, we formulate the recently proposed structural resolution as an abstract reduction system
%  (called \textit{LP-Struct}). We establish the following results in this paper:

%To match these three small-step operational semantics,
%introduce big-step operational semantics of resolution,
We give a type-theoretic semantics to all these reduction systems.
%This paper can be seen as a small step toward this goal: we propose the local
%productivity and give a precise definition. 
%This paper also presents an attempt at giving an analysis of several resolution
%strategies within a single coherent framework.
%In particular, we
Notably, we take System $\mathbf{H}$ (based on Howard's work \cite{ho80}) as a calculus
to capture the type-theoretic meaning of logic programming.
%to define the proof-theoretic meaning of logic programming.
%This single proof-theoretic framework then serves us as a solid foundation on which we build the comparative study of
%the different resolution strategies.
We show that LP-Unif, partial LP-Unif, LP-TM and LP-Struct are sound
relative to System $\mathbf{H}$. Moreover, LP-Unif is complete relative to System $\mathbf{H}$, and under a meaning preserving transformation,
LP-Struct is also complete relative to System $\mathbf{H}$. 
%% We have not proven the properties of infinite reductions by either of these systems, but we refer the reader to~\cite{FKSP15} for some preliminary results in that direction.
%and under which conditions, the three kinds of resolution are (inductively) sound and complete relative to System $\mathbf{H}$.
%In particular,

We discover that, given a program $\Phi$ and a formula $A$,
LP-Struct is operationally equivalent to LP-Unif %(which is sound and complete with respect to $\mathbf{H}$)
under
two conditions: when  all LP-Struct derivations for $\Phi$ are observaionally productive and when all clauses in $\Phi$ are non-overlapping.
Thus the termination of LP-TM plays a crucial role not only in ensuring observational productivity, but also in ensuring the operational equivalence of LP-Struct and LP-Unif,
which in its turn is crucial for our proofs of soundness and completeness of LP-Struct with respect to $\mathbf{H}$.
%Seeing the special role that LP-TM and its termination play in the analysis of properties of resoluton,

We give a formal analysis of properties of LP-TM resolution. We show how LP-TM relates to
term rewriting systems by introducing a transformation method that translates any logic
program without existential variables into a term rewriting system
(we call this process \emph{functionalisation}). After functionalisation, 
standard term rewriting methods for detecting termination can be applied. We also give an
alternative transformational method that renders all logic programs LP-TM terminating and
non-overlapping. The method is related to Kleene's realizability method \cite{KleeneSC:1952}, and we therefore call it \emph{realizability transformation}.

%% We see this study as a contribution to the broader effort of several communities (LP, model-checking, SMT-solving~\cite{Bj15}) to understand
%% operational properties of resolution for Horn clause logic, its terminating and nonterminating computation~\cite{LeinoM14,ReynoldsB15,SimonBMG07}.  

The technical content of this paper is organized as follows.

\begin{itemize}
   
\item In Section \ref{pre}, we prove soundness and completeness of LP-Unif with
  respect to the type system $\mathbf{H}$. This means $\mathbf{H}$ can be used
  to model logic programming. %, and that the usual notion of failure in Horn clause logic
  %can be understood as proving an implicational formula. 
  %  we extend the soundness and completeness results to LP-Struct

  \item In Section \ref{partial}, we formally define partial resolution as an abstract reduction system and call it
    %we provide a mechanism based on labels to control LP-Unif reductions such that they can perform lazy computation. The resulting resolution is called
    \textit{partial LP-Unif}. Based on this formalism we define local productivity. 
  Partial LP-Unif provides a possibility of shifting the focus
from deciding entailment for a given query to computing substitution answers.
  %on logical property (yes/no) of LP to its computational property (obtaining the actual answer). 
  
\item In Section \ref{rt:s}, we formally define LP-Struct  %, which is a refinement of LP-Unif by
  %a separation of unification and term-matching.
  and identify two conditions that ensure that
  LP-Struct is operationally equivalent to LP-Unif. %We prove soundness and completeness of LP-Struct (the latter property is proven modulo some conditions as explained above).

\item In Section \ref{s:func}, we define  \emph{functionalisation} %that transforms
 % a logic program without existential variables into term rewriting system, %we
  and show the exact relation of LP-TM to term rewriting systems. We use existing termination detection techniques from term rewriting to
  detect termination of LP-TM. 
  
\item In Section \ref{s:real}, we define \textit{realizability transformation} %that renders all logic programs observational productive. We
  and show that this transformation
      preserves the operational meaning of a logic program. We use it to show the equivalence of LP-Struct
    and LP-Unif for the transformed program.  As a corollary, we obtain the soundness and completeness of LP-Struct relative to System $\mathbf{H}$ for the transformed program.

    %% On the other hand, the nontermination of LP-TM can also post challenge to the current termination detection techniques in term rewriting, we give one such example. 
 
\end{itemize}

 Finally, in Sections~\ref{rw} and~\ref{concl} we survey related work and conclude the paper.

%%  Ofcourse,  a  number of
%% studies of terminating and nonterminating logic programs do already
%% exist~\cite{deSchreye1994199,GuptaBMSM07,SimonBMG07,Pf92,RohwedderP96}. We see the main contribution of this work to this existing body of knowledge
%% in setting up a unform proof-theoretic framework for analysis of different kinds of resolution, and different properties of operational semantics of LP.

%% *Perhaps a more succinct emphasis on the infinite data structures should be made in the intro *

\section{Horn Formulas as Types}
\label{pre}

In this section, we use Howard's type system $\mathbf{H}$ to model logic programming. 
We use an abstract reduction system (called LP-Unif) to model the small-step semantics of the SLD-resolution. The purpose of this section is
to set up a type-theoretic framework for the rest of the paper, where Horn formulas are viewed as types in a type system and resolution corresponds to the proof construction. 
We show the correspondence between the small-step semantics of resolution and System $\mathbf{H}$. This result can be viewed  as an alternative
to the classical-style soundness and completeness
 results for the SLD-resolution relative to Herbrand models. Using System \textbf{H} as an alternative
semantics for logic programming may be beneficial in two ways: (1) 
the usual notion of an unsuccessful SLD-derivation can be understood as proving an implicative formula in which the unresolved subgoals comprise the antecedent 
(see Lemma \ref{sound}). (2) It allows further extensions such as adding fixpoint typing rule by Fu \emph{et al}~\cite{flops16},  
which provides proofs for some nonterminating computations.

\begin{definition}[Syntax]

\

  Term $t \ ::= \ x \ | \ K(t_1,..., t_n)$

  Atomic Formula $A, B, C, D\ ::= \ P(t_1,...,t_n)$

  Formula $F \ ::=  A\ | \ F \Rightarrow F' \ | \ \forall x. F$

  Horn Formula/Horn Clause $H \ ::=  \forall \underline{x} . A_1,...,A_n \Rightarrow B$

  Proof Evidence $p, e \ ::= \ \kappa \ | \ a \ | \ e\  e' \ | \ \lambda a . e$

  Axioms/Logic Programs $\Phi \ ::= \cdot \ | \ \kappa :  H, \Phi \ | \ a : F, \Phi$
\end{definition}
Proof evidence is given by lambda terms. We use capitalised words to denote function symbols.  Constant evidence is denoted by $\kappa$.    %The $P$ in $P(t_1,...,t_n)$ is called \textit{predicate}. Predicate of arity zero is called \textit{proposition}. %LP formula $A_1,..., A_n \Rightarrow A$ can be intuitively viewed as $A_1 \to ...\to A_n \to A$, and the order of $A_1,..., A_n$ does not matter.
We write $A_1,..., A_n \Rightarrow B$ as a short hand for $A_1 \Rightarrow ... \Rightarrow A_n \Rightarrow B$. We write $\forall \underline{x} . F$ for quantifying over all free term variables in $F$, and $[\forall x].F$ denotes $F$ or $\forall x . F$. We use $\underline{A}$ to denote $A_1,..., A_n$, when the number $n$ is unimportant. If $n$ is zero for $\underline{A} \Rightarrow B$, then we write $\Rightarrow B$.  %LP program $B \Leftarrow \underline{A}$ are represented as 
Horn clause formulas have the form $\forall \underline{x} . \underline{A} \Rightarrow B$, and queries are given by  atomic formulas. We use $\mathrm{FV}(t)$ to denote the set of all free term variables in $t$. 

% We write $A \mapsto_\sigma A'$ (resp. $A \sim_\gamma A'$ ) to mean $A$ is matchable (resp. unifiable) to $A'$ with substitution $\sigma$ (resp. $\gamma$), i.e. $\sigma A \equiv A'$ (resp. $\gamma A \equiv \gamma A'$).

The following is a type system based on Howard's work \cite{ho80}, intended to provide a type theoretic interpretation for LP.  
\begin{definition}[Howard's System $\mathbf{H}$ for logic programming]
  \label{proofsystem}

\

{
\begin{tabular}{lllll}
\\
\infer[\textsc{Axiom}]{\Phi \vdash \kappa  : H}{(\kappa : H) \in \Phi}    
&

&

\infer[\textsc{Var}]{\Phi \vdash  a : F}{(a : F) \in \Phi}

&

&
\infer[\textsc{App}]{\Phi \vdash e_1 \ e_2 : F_2}{\Phi \vdash e_1 : F_1 \Rightarrow F_2 & \Phi \vdash e_2 : F_1}
\\
\\
\infer[\textsc{Inst}]{\Phi \vdash e : [t/x]F}{\Phi \vdash e : \forall x . F}

&

&
\infer[\textsc{Gen}]{\Phi \vdash  e: \forall x . F}{\Phi \vdash e : F}

& &

\infer[\textsc{Abs}]{\Phi \vdash \lambda a. e : F_1 \Rightarrow F_2}{\Phi, a : F_1 \vdash e : F_2}
  \end{tabular}
}  
\end{definition}

Note that the type for the constant in the rule \textsc{Axiom} is required to be Horn formula. 
 It has been observed that the \textsc{Cut} rule and proper axioms in intuitionistic sequent calculus can emulate logic programming~\cite{Girard:1989}(\S 13.4). 

The following rule is a version of \textsc{Cut} rule, working only with Horn formulas.  

\begin{center}
  \infer[\textsc{Cut}]{\Phi \vdash \lambda \underline{a}. \lambda
    \underline{b}. (e_2 \ \underline{b})\ (e_1\ \underline{a}) :
    \underline{A}, \underline{B} \Rightarrow C}{\Phi \vdash e_1 :
    \underline{A} \Rightarrow D & \Phi \vdash e_2 : \underline{B}, D
    \Rightarrow C}
\end{center}
We can use rules \textsc{Abs} and \textsc{App} to emulate \textsc{Cut} rule, thus the \textsc{Cut} rule is admissible in Howard's system $\mathbf{H}$. We will use $\mathbf{C}$
to denote the deduction system that consists of rules \textsc{Axiom}, \textsc{Cut}, \textsc{Inst}, and \textsc{Gen}. 
The subsystem $\mathbf{C}$ provides a natural framework to work with Horn formulas, but $\mathbf{H}$ is more expressive, since it allows full intuitionistic formulas, e.g. $\mathbf{H}$ would allow a formula of the form $(F_1 \Rightarrow F_2) \Rightarrow F_3$. %, but not $\mathbf{C}$. 

\begin{definition}
  Beta-reduction on proof evidence is defined as the congruence closure of the following relation:
 $(\lambda a . e)\ e' \to_\beta [e'/a]e$. We say a proof evidence $e$ is strongly normalizing if $e$ admits 
no infinite $\to_\beta$-reductions. 
\end{definition}

The following three theorems are standard for a type system such as \textbf{H}. For proofs
we refer the reader to Barendregt's excellent book \cite{barendregt1993}. 

\begin{theorem}[Strong Normalization]
\label{real:sn}
 If $\Phi \vdash e : F$ in $\mathbf{H}$, then $e$ is strongly realisable with respect to beta-reduction on proof evidence.
\end{theorem}

%% The proof of strong normalization (SN) is standard, similar to that of simply typed lambda calculus. The following two theorems are also standard for $\mathbf{H}$.

\begin{theorem}[Inversion]
  
  \begin{itemize}
  \item If $\Phi \vdash a : F$, then $(a : F') \in \Phi$ and $\sigma F' \equiv F$ for some substitution $\sigma$.
      \item If $\Phi \vdash \kappa : H$, then $(\kappa : \forall \underline{x} . A_1,..., A_n \Rightarrow B) \in \Phi$ and $\sigma (A_1,..., A_n \Rightarrow B) \equiv H$ for some substitution $\sigma$.
            \item If $\Phi \vdash \lambda a.e : F$, then $\Phi, a : F_1 \vdash e: F_2$ and $\sigma (F_1 \Rightarrow F_2) \equiv F$ for some substitution $\sigma$.
                    \item If $\Phi \vdash e_1\ e_2 : F$, then $\Phi \vdash e_1 : F_1 \Rightarrow F_2$, $\Phi \vdash e_2 : F_1$ and $\sigma F_2 \equiv F$ for some substitution $\sigma$.
        
  \end{itemize}
\end{theorem}

\begin{theorem}[Type Preservation]
System $\mathbf{H}$ is type preserving, i.e. if $\Phi \vdash e : F$ in \textbf{H} and $e \to_\beta e'$, then $\Phi \vdash e' : F$. 
\end{theorem}

Note that system $\mathbf{C}$  as a type system is not type preserving. For example, consider $\Phi = (\kappa_1 : A \Rightarrow B, \kappa_2 : B \Rightarrow C, \kappa_3 : C \Rightarrow D)$. In \textbf{C}, we have 
$\Phi \vdash \lambda a . \kappa_3 \ ((\lambda b . \kappa_2\ (\kappa_1\ b))\ a) : A \Rightarrow D$. But 
$\lambda a . \kappa_3 \ ((\lambda b . \kappa_2\ (\kappa_1\ b))\ a) \to_\beta \lambda a . \kappa_3 \ (\kappa_2\ (\kappa_1\ a))$ and $\Phi \not \vdash \lambda a . \kappa_3 \ (\kappa_2\ (\kappa_1\ a)) : A \Rightarrow D$ 
in \textbf{C}. 
 Thus we often 
work with $\mathbf{C}$ through its embedding in $\mathbf{H}$, which is type preserving and strongly 
normalising.   

\begin{definition}[Unification]
  We say that $t$ is unifiable with $t'$  with substitution $\gamma$ (denoted $t \sim_\gamma t'$),  if $\{t = t'\} 
\rightarrowtail^* \gamma$ according to the following rules:
 
\[
\begin{array}{lll}

\{K(t_1,...,t_n) = K(s_1,...,s_n)\} \cup E & \rightarrowtail & \{t_1 = s_1 ,...,t_n = s_n\} \cup E

\\
\\

\{K(t_1,...,t_n) = G(s_1,...,s_m)\} \cup E &\rightarrowtail & \bot

\\
\\

\{t = t \} \cup E & \rightarrowtail &  E

\\
\\

\{K(t_1,...,t_n) = x \} \cup E & \rightarrowtail&  \{x = K(t_1,...,t_n) \} \cup E

\\
\\

\{x = K(t_1,...,t_n) \} \cup E & \rightarrowtail&  \bot \ \text{if}~x \in \mathrm{FV}(K(t_1,...,t_n)) 

\\
\\

\{ x = t \} \cup E & \rightarrowtail& \{ x = t \} \cup [t/x]E \ \text{if} \ x \notin \mathrm{FV}(t)

  \end{array}
\]
\end{definition}

Unification can be routinely extended to atomic formulas. The symbol $\bot$ denotes failure of unification. The following is a formulation of the SLD-resolution as a reduction system, as given in Nilsson and Maluszynski \cite{nilsson1990logic}.

\begin{definition}[LP-Unif reduction]
\label{red}
Given a set of axioms $\Phi$, we define a reduction relation on the multiset of atomic formulas: 

\noindent $\Phi \vdash \{A_1,..., A_i, ..., A_n\} \leadsto_{\kappa, \gamma \cdot \gamma'} \{\gamma A_1,..., \gamma B_1,..., \gamma B_m, ..., \gamma A_n\}$ for any substitution $\gamma'$, if there exists $\kappa : \forall \underline{x} . B_1,..., B_n \Rightarrow C \in \Phi$ such that $C \sim_{\gamma} A_i$.
\end{definition}
The second subscript in the reduction is intended as a state, it will be updated by composition along with reductions. Notation $\gamma \cdot \gamma'$ should be read as follows:
the old state $\gamma'$ is updated, producing a new state $\gamma \cdot \gamma'$. We assume fresh names in the form of new numeric indices for the quantified variables each time the above rule is applied. We write $\leadsto$ when we leave the associated state implicit. We use $\leadsto^*$ to denote the reflexive and transitive closure of $\leadsto$. Notation $\leadsto_{\gamma}^*$ is used when the final state along the reduction path is $\gamma$.

Given a program $\Phi$ and
a set of queries $\{B_1, \ldots, B_n\}$, SLD-resolution uses LP-Unif reduction to reduce $\{B_1, \ldots, B_n\}$: 

\begin{definition}[LP-Unif]
\
\noindent  Given a logic program $\Phi$, LP-Unif is given by the abstract reduction system $(\Phi, \leadsto)$. 
\end{definition}

\begin{example}\label{ex:conn}
Consider the following logic program  $\Phi$, consisting of Horn formulas labelled by $\kappa_1$, $\kappa_2$, $\kappa_3$, defining connectivity for a graph with three nodes:

  \begin{center}
  $\kappa_1 : \forall x. \forall y. \forall z. Connect(x, y), Connect(y, z) \Rightarrow {Connect}(x, z)$

  $\kappa_2 : \ \Rightarrow {Connect}({Node_1}, {Node_2})$
  
  $\kappa_3 : \ \Rightarrow {Connect}({Node_2}, {Node_3})$
    \end{center}

The usual SLD-resolution for the query ${Connect}(x, y)$ can be represented as the following 
 LP-Unif reduction:

  \begin{center}
    $\Phi \vdash \{{Connect}(x, y)\} \leadsto_{\kappa_1, [x/x_1, y/z_1]}
    \{{Connect}(x, y_1), {Connect}(y_1, y)\} \leadsto_{\kappa_2, [{Node_1}/x, {Node_2}/y_1, {Node_1}/x_1, y/z_1]}
    \{{Connect}({Node_2}, y)\} \leadsto_{\kappa_3, [{Node_3}/y, {Node_1}/x, {Node_2}/y_1,
      {Node_1}/x_1, {Node_3}/z_1]} \emptyset $
  \end{center}

 \noindent The first reduction $\leadsto_{\kappa_1, [x/x_1, y/z_1]}$  unifies the query ${Connect}(x, y)$ with the head of the rule $\kappa_1$ (which is ${Connect}(x_1, z_1)$
 after renaming) with the substitution $[x/x_1, y/z_1]$ ($x_1$ is replaced by $x$ and $z_1$ is replaced by $y$).
So the query is \emph{resolved} with $\kappa_1$, 
  producing the next queries: ${Connect}(x, y_1)$, ${Connect}(y_1, y)$. Note that
  the substitution in the subscript of $\leadsto$ is a state that will be updated alongside the derivation. In the final state we have an answer $[{Node_3}/y, {Node_1}/x]$ for the query ${Connect}(x, y)$.

\end{example}

\subsection{Soundness and Completeness of LP-Unif}
We have introduced the Howard's system \textbf{H} and LP-Unif. 
Now we will show the soundness of LP-Unif, i.e., we show that a reduction by LP-Unif corresponds to an 
intuitionistic proof. On the other hand, any first order ground evidence of type $A$ in \textbf{H} corresponds to a successful LP-Unif reduction (which is the essence of the completeness result). 

\begin{lemma}[Soundness Lemma]
\label{sound}
  If $\Phi \vdash \{A\} \leadsto^*_{\gamma} \{B_1,..., B_n\}$, then $\Phi \vdash e : B_1,..., B_n \Rightarrow \gamma A$ for some $e$ in $\mathbf{C}$. 
\end{lemma}

\begin{proof}
  By induction on the length of the reduction.
  
\begin{itemize}
\item Base Case. Suppose the length is one, namely, $\Phi \vdash \{A\} \leadsto_{\kappa, \gamma} \{B_1,..., B_n\}$. It implies $(\kappa : \forall \underline{x} . B_1',..., B_n' \Rightarrow C) \in \Phi$, $\gamma B_i' \equiv B_i$ and $C \sim_\gamma A$.  So we have $\Phi \vdash \kappa : \gamma B_1',..., \gamma B_n' \Rightarrow \gamma C$ by the rules \textsc{Axiom} and \textsc{Inst}. 

\item Step Case. Suppose
 $\Phi \vdash \{A\} \leadsto_{\gamma_1}^* \{A_1, ..., A_i,..., A_n\} \leadsto_{\kappa, \gamma_2\cdot \gamma_1} \{\gamma_2 A_1,..., \gamma_2 B_1,..., \gamma_2 B_m,..., \gamma_2 A_n\}$, where $\kappa : \forall \underline{x} . B_1,..., B_m \Rightarrow C$ and $C \sim_{\gamma_2} A_i$. By inductive hypothesis (IH), we have $\Phi \vdash e_1 : A_1,..., A_n \Rightarrow \gamma_1 A$. 
By \textsc{Inst} and \textsc{Gen}, 
we have $\Phi \vdash e_1 : \gamma_2 A_1,..., \gamma_2 A_i,..., \gamma_2 A_n \Rightarrow \gamma_2 \gamma_1 A$ and
 $\Phi \vdash \kappa : \gamma_2 B_1,..., \gamma_2 B_m \Rightarrow A_i$. Since $\gamma_2$ is idempotent, we 
have $\Phi \vdash \kappa : \gamma_2 B_1,..., \gamma_2 B_m \Rightarrow \gamma_2 A_i$. Thus 
by \textsc{Cut} rule, we have $\Phi \vdash e' :  \gamma_2 A_1,..., \gamma_2 B_1,..., \gamma_2 B_m,..., \gamma_2 A_n \Rightarrow \gamma_2 \gamma_1 A$ for some proof evidence $e'$.

\end{itemize}
\end{proof}

The soundness lemma above ensures that \textit{every} LP-Unif reduction and its answer 
are meaningful. The usual notion of failure in logic programming can be understood as proving an implicative
formula in which the antecedent is comprised of the failed subgoals. The notion of success corresponds to a proof of an atomic formula. For example, consider
the logic program $\Phi = \kappa_1 : P_3(K) \Rightarrow P_1(K), \kappa_2 : P_2(K) \Rightarrow P_3(K)$. We know that the query $P_1(x)$ will fail. But by Lemma \ref{sound}, we know the resolution for query $P_1(x)$ will stop at $P_2(K)$ with substitution $[K/x]$, and we 
have the proof $\Phi \vdash \lambda a. \kappa_1 \ (\kappa_2 \ a) : P_2(K) \Rightarrow [K/x]P_1(x)$ in 
\textbf{C}. 
So in a sense the failed query $P_1(x)$ is still meaningful under the type theoretic interpretation.  
%% This interpretation may be useful
%% since some LP-Unif reductions can be nonterminating.
In Section \ref{partial} we will use LP-Unif in a way that it does not have to resolve all the queries, but it still
computes useful answers for the variable that we care about. 

\begin{theorem}[Soundness of LP-Unif]
\label{sound:unif}
    If $\Phi \vdash \{A\} \leadsto^*_\gamma \emptyset$ , then $\Phi \vdash e : \forall \underline{x} .\Rightarrow \gamma A$ for some $e$ in $\mathbf{C}$. %% and $e$ is a first-order proof evidence, furthermore, $e$ is in beta normal form.  
\end{theorem}

An evidence is \textit{ground} if it does not contain free evidence variables. The proof of completeness relies on the strong normalisation and the type preservation property of \textbf{H}. We first show
that the normal form of the proof given by a successful LP-Unif reduction is first-order. We then show that, if an atomic formula is inhabited by a ground evidence, there exists a successful LP-Unif reduction for it. 

\begin{definition}[First-Order Proof Evidence]
  We define first-order proof evidence as follows. 
  \begin{itemize}
  \item A variable proof evidence  $a$ and a constant proof evidence  $\kappa$ are first-order.
  \item If $n, n'$ are first-order, then $n\ n'$ is first-order.
  \end{itemize}
\end{definition}
For example, $(\kappa \ \kappa')$ is considered first-order, but $(\kappa \ (\lambda a . \kappa' \ (\kappa'' \ a)))$
is not first-order.  
\begin{proposition}
  \label{fo:sub}
  If $n, n'$ are first-order, then $[n'/a]n$ is first-order. 
\end{proposition}

\begin{lemma}
\label{fst:lambda}
   If $\Phi \vdash e : [\forall \underline{x}.] \underline{A} \Rightarrow B$ in $\mathbf{C}$, then either $e$ is a proof evidence constant, variable or it is normalisable to the form $\lambda \underline{a}. n$, where $n$ is first-order normal proof evidence. 
\end{lemma}
\begin{proof}
  By induction on the derivation of $\Phi \vdash e : [\forall \underline{x}.] \underline{A} \Rightarrow B$.
  \begin{itemize}
  \item Base Case. Rule \textsc{Axiom}. Obvious.
    \item Step Case.

\
      
           \begin{tabular}{l}
\infer[\textsc{Cut}]{\Phi \vdash \lambda \underline{a} . \lambda \underline{b} . (e_2\ \underline{b})\ (e_1\ \underline{a}) : \underline{A}, \underline{B} \Rightarrow C}{\Phi \vdash e_1 : \underline{A} \Rightarrow D & \Phi \vdash  e_2 : \underline{B}, D \Rightarrow C}
\end{tabular}

\

           By IH, we know that $e_1 = \kappa$ or $e_1 = \lambda \underline{a}.n_1$; $e_2 = \kappa'$ or $e_2 = \lambda \underline{b}.\lambda d . n_2$, where $n_1, n_2$ are fist-order. 
 We know that $e_1\ \underline{a}$ will be normalizable to a first-order proof evidence. And $e_2 \ \underline{b}$ will be normalized to either $\kappa'\ \underline{b}$ or $\lambda d . n_2$. So by Proposition \ref{fo:sub}, we conclude that 
           $\lambda \underline{a} . \lambda \underline{b} . (e_2\ \underline{b})\ (e_1\ \underline{a})$ is normalizable to $\lambda \underline{a} . \lambda \underline{b} . n$ for some
           first-order normal term $n$.
           \item The \textsc{Gen} and \textsc{Inst} cases are straightforward.
  \end{itemize}
\end{proof}
\begin{theorem}
  \label{fst}
    If $\Phi \vdash e : [\forall \underline{x}.] \Rightarrow B$ in $\mathbf{C}$, then $e$ is normalizable to a first-order proof evidence.
\end{theorem}

Now let us prove the completeness theorem. %% The following lemma shows that completeness result holds for the transformed program. 

\begin{proposition}
  If $e$ is a ground first-order evidence, then it is of the following form: 
  \begin{itemize}
  \item $\kappa$
    \item $\kappa \ n_1\ ...\ n_l$, where $n_i$ is ground first-order evidence for any $i$. 
  \end{itemize}
\end{proposition}

\begin{theorem}[Completeness of LP-Unif]
If $\Phi \vdash n : \ \Rightarrow A$ where $n$ is in ground first-order normal form in $\mathbf{H}$, then $\Phi \vdash \{A\} \leadsto^* \emptyset$.
\end{theorem}

\begin{proof}
  By induction on the structure of $n$.
  \begin{itemize}
  \item Base Case. $n = \kappa$. By inversion, we know 
$\kappa : \forall \underline{x}. \ \Rightarrow A' \in \Phi$ and $\gamma A' \equiv A$ for some substitution $\gamma$. Thus $A' \sim_\gamma A$, which implies $\Phi \vdash \{A\} \leadsto_{\kappa, \gamma} \emptyset$.
  \item Step Case. $n = \kappa\ n_1\ n_2\ ...\ n_m$. By inversion, 
we have $\kappa : \forall \underline{x}. \ C_1,..., C_m\Rightarrow B \in \Phi$. 
    To obtain $\Phi \vdash n : \ \Rightarrow A$, by inversion we have $\Phi \vdash \kappa : \forall \underline{x}. \ C_1,..., C_m\Rightarrow B$ with $\gamma_m ... \gamma_1 (B) \equiv A$, and
     $\Phi \vdash n_1 :\ \Rightarrow C_1, \Phi \vdash n_2 :\ \Rightarrow \gamma_1 C_2 ..., \Phi \vdash n_m : \ \Rightarrow \gamma_{m-1}... \gamma_1\ C_m$. By the rule \textsc{Inst}, we have
         $\Phi \vdash n_1 :\ \Rightarrow \gamma_{m}... \gamma_1 C_1, \Phi \vdash n_2 :\ \Rightarrow \gamma_{m}... \gamma_1 C_2 ..., \Phi \vdash n_m : \ \Rightarrow \gamma_{m}... \gamma_1\ C_m$. 
 Thus we have $\Phi \vdash \{A\} \leadsto_{\kappa, \gamma_m \cdot ... \cdot \gamma_1} \{\gamma_m ... \gamma_1 C_1,..., \gamma_m ... \gamma_1 C_m\}$. By IH, 
we have $\Phi \vdash \{\gamma_{m}... \gamma_1 C_1\} \leadsto^*_{\sigma_1} \emptyset$. So $\Phi \vdash \{A\} \leadsto_{\kappa, \gamma_m \cdot ... \cdot \gamma_1} \cdot \leadsto^*_{\sigma_1} \{\sigma_1 \gamma_{m}... \gamma_1 C_2,..., \sigma_1 \gamma_{m}... \gamma_1 C_m \} $. Again, we have $\Phi \vdash n_2 :\ \Rightarrow \sigma_1 \gamma_m ... \gamma_1 C_2$ by rule \textsc{Inst}. By applying IH repeatedly, we obtain $\Phi \vdash \{A\} \leadsto^* \emptyset$.
  \end{itemize}
\end{proof}

%% \subsection{Discussion}
%% There have been many works on understanding logic programming using proof system. Girard pointed out \cite{Girard:1989} that atomic intuitionistic sequent calculus with \textsc{Axioms} and \textsc{Cut} rules
%% are the only rules needed to represent logic programming. Miller et. al. \cite{Miller:1989} use cut-free sequent
%% calculus and its \textit{uniform proof} to understand logic programming. The uniform proof in Miller et. al.'s work are sequent calculus proofs that admit nice inversion property. 

%% The difference between previous works and our formulation in this section appears to be superficial. Instead of using sequent calclus as the proof theoretic framework, we use type system with
%% proof term annotation as our framework, and we formulate resolution
%% as a reduction system that can generate proof term. This choice allows further refinment and extension on both resolution and
%% the type system: for example, we can extend the type system with fixpoint operator \cite{flops16} to cope with certain coinductive behavior, or we can further refine the reduction system to 
%% develop different reduction strategies as we will see in the following sections. 

\section{Partial LP-Unif by Labelling}
\label{partial}
In the previous section, we have given a type-theoretic semantics to logic programming. According to it,  an answer for a given
query is a substitution applied to a formula that is inhabited by a proof evidence. 
In that sense, the soundness lemma (Lemma \ref{sound}) gives type-theoretic meaning
to \textit{any} LP-Unif reduction, even if it is a partial derivation, i.e. 
has unresolved subgoals. % refutation in the standard LP sense. % from a query is meaningful. 
 %In logic programming, sometimes one may want
In this section, we build upon this result, and propose a \textit{lazy} version of LP-Unif,
drawing inspiration from lazy functional languages such as Haskell.
%% The method we propose here is a lightweight approach to performing lazy computation
%% in logic programming. 
In particular, we propose to label certain variables in a given query, 
in order to prioritize those variables for which
we want to compute substitutions. 
%the intention of the need of the answers for the labelled variables.
Partial LP-Unif resolution will only  
resolve the subgoals that contain labelled variables. This requires to extend the usual unification to account for labels.
We call the resulting unification algorithm \textit{labelled unification} and the resulting reduction strategy -- \textit{partial LP-Unif}.

\begin{definition}
  We extend the term definition:\\
  $t \ ::= x \ | \ x^v \ |\ K(t_1,...,t_n)$,
  where $x^v$ is a labelled variable.\\ Definitions of a Horn formula and a formula are  extended accordingly. 
\end{definition}

A label on a variable can be informally understood as a case-expression on a variable in lazy functional language. When a query has a labelled
variable, it forces  resolution to compute a value for it. But since we are in logic programming, the only
way to force such evaluation is through label propagation and elimination. The following definition
extends unification to achieve this. % account for labels. 

We write $t^v$ to denote the labelled version of $t$, in which all the variables of $t$ are labelled. Note that 
 $x^v$ is identical to $(x^v)^v$. 

\begin{definition}[Labelled Unification]
  We say that $t$ is unifiable with $t'$  with substitution $\gamma$ (denoted $t \sim_\gamma t'$),  if $\{t = t'\} 
\rightarrowtail^* \gamma$ according to the following rules:
  %We define $t \simeq_\gamma t'$ (\emph{$t$ is unifiable with $t'$ with substitution $\gamma$}), if $\{t = t'\} 
%\rightarrowtail^* \gamma$. The relation $\rightarrowtail$ is defined as following.  

\[
\begin{array}{lll}

\{K(t_1,...,t_n) = K(s_1,...,s_n)\} \cup E & \rightarrowtail & \{t_1 = s_1 ,...,t_n = s_n\} \cup E

\\
\\

\{K(t_1,...,t_n) = G(s_1,...,s_m)\} \cup E &\rightarrowtail & \bot

\\
\\

\{t = t \} \cup E & \rightarrowtail &  E

\\
\\

\{K(t_1,...,t_n) = x \} \cup E & \rightarrowtail&  \{x = K(t_1,...,t_n) \} \cup E

\\
\\

\{K(t_1,...,t_n) = \remph{x^v} \} \cup E & \rightarrowtail&  \{\remph{x^v} = K(t_1,...,t_n) \} \cup E

\\
\\

\{x = K(t_1,...,t_n) \} \cup E & \rightarrowtail&  \bot \ \text{if}~x \in \mathrm{FV}(K(t_1,...,t_n)) 

\\
\\

\{\remph{x^v} = K(t_1,...,t_n) \} \cup E & \rightarrowtail&  \bot \ \text{if}~\remph{x^v} \in \mathrm{FV}(K(t_1,...,t_n))

\\
\\

\{ x = t \} \cup E & \rightarrowtail& \{ x = t \} \cup [t/x]E \ \text{if} \ x \notin \mathrm{FV}(t)

\\
\\

\{ \remph{x^v} = t \} \cup E & \rightarrowtail& \{ \remph{x^v} = \remph{t^v} \} \cup [\remph{t^v/x^v}](\ell_t (E)) \ \text{if} \ x^v \notin \mathrm{FV}(t)
  \end{array}
\]
\end{definition}
We use $\ell_t(E)$ to denote a labelling operation that labels all the variables in $E$ that occur in $t$. Formally, $\ell_t(E)$ is defined as $\sigma E$, where $[x^v/x] \in \sigma$ for any $x \in \mathrm{FV}(t)$. The set of equations $\{x_1 = t_1, ..., x_n = t_n\}$ can be viewed as a substitution. The labelled unification of terms can be extended routinely to the unification of atomic formulas. 

We write $|t|$ to denote erasing all the labels in $t$. We write $|\sigma|$ to
denote removing all the labels in the substitution $\sigma$, and $L(A)$ to denote the set of labelled variables in $A$. The following lemma shows  that labelled unification is functionally equivalent
to the usual (unlabelled) unification. 

 \begin{lemma}
   \label{la-unif}
   If $t \simeq_{\gamma} t'$ and $\mathrm{FV}(|t|) \cap \mathrm{FV}(|t'|) = \emptyset $, then $|t| \sim_{|\gamma|} |t'|$.
 \end{lemma}
We use $\ell_\gamma (A)$ to denote another labelling operation that labels all variables in $A$
that are labelled in the codomain of $\gamma$. Formally, $\ell_\gamma (A)$ is defined as $\sigma A$, such that, for any  $x^v \in \mathrm{FV}(\mathrm{codom}(\gamma))$, $[x^v/x]\in \sigma$.
\begin{definition}[Partial LP-Unif]

We define a reduction relation on a multiset of atomic formulas: 

\noindent $\Phi \vdash \{A_1,..., A_i, ..., A_n\} \rightharpoonup_{\kappa, \gamma \cdot \gamma'} \{\gamma A_1,..., \gamma \ell_{\gamma}(B_1),..., \gamma \ell_{\gamma}(B_m), ..., \gamma A_n\}$ for any substitution $\gamma'$, if $L(A_i) \not = \emptyset$ and there exists $\kappa : \forall \underline{x} . B_1,..., B_n \Rightarrow C \in \Phi$ such that $C \simeq_{\gamma} A_i$.

\end{definition}

The labelling operation $\ell_\gamma$ is used in the above definition to make sure that 
the labels are correctly propagated to $B_1,..., B_m$. Consider $\Phi = \kappa_1 :\ \Rightarrow P_1(Int), \kappa_2 :\ \Rightarrow P_2 (Int), \kappa_3 : \forall x . P_1(x), P_2(x) \Rightarrow Q(List(x))$, and the query $Q(z^v)$. The labelled unifier
of $Q(z^v)$ and $Q(List(x))$ is $\gamma = [List(x^v)/z^v]$, we need to apply $\ell_\gamma$ to $P_1(x)$ and $P_2(x)$ to obtain the resolvents $P_1(x^v)$ and $P_2(x^v)$.

As a consequence of Lemma \ref{la-unif}, the partial LP-Unif is essentially a control
strategy for LP-Unif.

\begin{lemma}\label{spu}
  If $\Phi \vdash \{A_1,..., A_i, ..., A_n\} \rightharpoonup_{\kappa, \gamma \cdot \gamma'} \{\gamma A_1,..., 
\gamma \ell_\gamma(B_1),..., \gamma \ell_\gamma(B_m), ..., \gamma A_n\}$, then

\noindent $\Phi \vdash \{|A_1|,..., |A_i|, ..., |A_n|\} \leadsto_{\kappa, |\gamma \cdot \gamma'|} \{|\gamma A_1|,..., |\gamma \ell_\gamma(B_1)|,..., |\gamma \ell_\gamma(B_m)|, ..., |\gamma A_n|\}$. 
\end{lemma}

Note that the above lemma implies, as a corollary, that partial LP-Unif is sound with respect to $\mathbf{H}$. 

\begin{definition}
  If $\Phi \vdash \{A\} \rightharpoonup_\gamma^* \{B_1,..., B_n\}$, $L(A) \not = \emptyset$ and $L(B_i) = \emptyset$ for all $i$, then we say $\gamma$ is the \textit{relative answer} for the labelled variables in $A$ with respect to $B_1,..., B_n$. 
\end{definition}

Comparing to LP-Unif, partial LP-Unif does not resolve subgoals without labelled variables, and therefore it terminates as soon as  all the labelled formulas
are resolved.
From the pragmatic perspective, the termination of partial LP-Unif signifies that 
%For partial LP-Unif, if all the labels are eliminated, then
we have obtained
all the answers we need for the labelled variables, and thus no further computation is necessary.
From the computational perspective, %there is no need to
 %continue to perform computation on the query that does not contain labels because: 1. The answers
 %for such query is not needed. 2.
some queries may give rise to nonterminating LP-Unif reduction and the partial LP-Unif strategy
offers a lazy version of LP-Unif as an alternative. % execution provides an alternative notion of productive computation.
%we know that 
This is a useful lightweight solution, as checking (non)termination for logic programming
can be at best only semi-decidable.
%nontermination for resolution in general is undecidable. 
From the theoretical perspective, relative answers are meaningful according
to type-theoretic semantics of Section \ref{pre} via Lemma~\ref{spu}.

Labels on variables give us a precise way to formalize the notion of local productivity we mentioned
in Introduction. 
\begin{definition}[Local Productivity]
    \label{local}

We say the queries $A_1,..., A_n$ are \textit{locally productive} at the set 
 of labelled variables $V$ iff $(\{A_1, ..., A_n\}, V) \in \mathrm{LProd}_m$ for any $m \geq 0$. We define $(\{A_1, ..., A_n\}, V) \in \mathrm{LProd}_m$ as follows:
 \begin{itemize}
 \item $(\{A_1, ..., A_n\}, V) \in \mathrm{LProd}_0$ for any $A_1, ..., A_n, V$.
 \item $(\{A_1, ..., A_n\}, V) \in \mathrm{LProd}_{m+1}$ iff for all $x^v \in V$, there exists
  $\Phi \vdash \{A_1,..., A_n \} \rightharpoonup^*_{\gamma}  \{Q_1,..., Q_l\}$ such
   that $\gamma(x^v) = K(t_1,...,t_n)$ and $(\{Q_1,..., Q_l\}, L(\gamma)) \in \mathrm{LProd}_m$, where $L(\gamma)$ denotes the set of labelled variables in the codomain of $\gamma$.
 \end{itemize}
 
\end{definition}

Note that according to the definition, if a query is trivially locally productive at the
empty set of labelled variables and is trivially locally unproductive if the set of labelled
variables are not a subset of the labelled variables in the query. Thus in general it is
more sensible to ask local productivity of a query at all its labelled variables.  

Local productivity is defined by 
quantifying over all natural numbers $m$ on $\mathrm{LProd}_m$, this allows us to prove
local productivity by induction on natural numbers. The requirement $\gamma(x^v) = K(t_1,...,t_n)$ in the definition of $\mathrm{LProd}_m$
ensures that the answer for a labelled variable $x^v$ is at least observable at the function 
symbol $K$. 

To illustrate local productivity, let us consider the query $P(x^v)$ and the logic program
 $\kappa : \forall x. P(x) \Rightarrow P(K(x))$. We want to show that $\{P(x^v)\}$ is locally
productive at $x^v$ for any variable $x$. We just need to show 
for any $x$, $(\{P(x^v)\}, \{x^v\})\in \mathrm{LProd}_m$ for all $m$. We proceed by induction on $m$. Suppose $m = 0$, we know that $(\{P(x^v)\}, \{x^v\}) \in \mathrm{LProd}_0$. Suppose $m = m'+1$, we have a partial LP-Unif reduction $\Phi \vdash \{P(x^v)\} \rightharpoonup_{[K(x_1^v)/x^v]} \{P(x_1^v)\}$. As $[K(x_1^v)/x^v]x^v = K(x_1^v)$, we just need to show 
$(\{P(x_1^v)\}, \{x_1^v\}) \in \mathrm{LProd}_{m'}$, 
which is by the inductive assumption. This is a very simple case of showing local productivity. In general, it is very challenging to prove local productivity in advance, and we leave 
this to future work.

Partial productivity and global productivity do not coincide, in general. For example, for a logic program $\kappa:  P(K) \Rightarrow P(K) $, the query $P(x^v)$ is locally productive, but it is not globally productive,
as, using terminology of Lloyd~\cite{Llo88}, no infinite term gets produced at infinity.

We give two further examples of performing finite computation on infinite data structures
using partial LP-Unif. 

\begin{example}
    \label{from}
Consider the following logic program $\Phi$, that observes, via the $Nth$ predicate, elements of an infinite stream of successive integers defined by $From$: 

\begin{center}

$\kappa_1 : \forall x . \forall y . From(S(x), y) \Rightarrow From(x, Cons(x,y))$

\

$\kappa_2 : \forall x . \forall y . \Rightarrow Nth(Z, Cons(x, y), x)$

\

$\kappa_3 : \forall x . \forall y . \forall z . \forall u .  Nth(x, z, u) \Rightarrow Nth(S(x), Cons(y, z), u)$
\end{center}
For the query $Nth(S(Z), y, \remph{z^v}), From(S(Z), y)$, we only want to know the answer for 
$\remph{z^v}$, i.e., the 2nd element in the stream generated by $From(S(Z), y)$. We observe the following reduction: 

\begin{center}
{
  $\Phi \vdash \{Nth(S(Z), y, \remph{z^v}), From(S(Z), y)\}
  \rightharpoonup_{\kappa_3, \gamma_1 \equiv [Z/x_1, Cons(y_1, z_1)/y, \remph{u_1^v/z^v}]} \{Nth(Z, z_1, \remph{u_1^v}), From(S(Z), Cons(y_1, z_1))
  \} \rightharpoonup_{\kappa_2, \gamma_2 \equiv [\remph{x_2^v/u_1^v}, Cons(\remph{x_2^v}, y_2)/z_1]\cdot \gamma_1} 
\{From(S(Z), Cons(y_1, Cons(\remph{x_2^v}, y_2) ))\} \rightharpoonup_{\kappa_1, \gamma_3 \equiv [Cons(\remph{x_2^v}, y_2)/y_3,S(Z)/y_1,S(Z)/x_3]\cdot \gamma_2} \{From(S(S(Z)), Cons(\remph{x_2^v}, y_2))\}\rightharpoonup_{\kappa_1, [y_4/y_2, \remph{S(S(Z))/x_2^v} , S(S(Z))/x_4]\cdot \gamma_3} \{From(S(S(S(Z))), y_4)\}
$
}
\end{center}
Thus $S(S(Z))$ is the answer for $z^v$ relative to $From(S(S(S(Z))), y_4)$, i.e. the 2nd element in the stream generated by $From(S(Z), y)$ is $S(S(Z))$. 
\end{example}

\begin{example}
  \label{fib}
Consider the following logic program $\Phi$: 

  \begin{center}
    $\kappa_1 : \forall x . \forall y . \Rightarrow Take(Z,App(x,y),Nil)$

    \

    $\kappa_2 : \forall x . \forall y . \forall z . \forall r . Take(x,z,r) \Rightarrow Take(S(x),App(y,z),Cons(y,r))$

    \

    $\kappa_3 : \forall x . \forall y . \forall s .Fib(y,App(x, y),s) \Rightarrow Fib(x,y,App(x,s)) $
  \end{center}
The formula $Fib(y,App(x, y),s) \Rightarrow Fib(x,y,App(x,s))$ is intended to
 generate (potentially infinitely long) Fibonacci word. For example, $A, B, A\cdot B, B\cdot (A\cdot B), 
 (A\cdot B)\cdot (B\cdot (A\cdot B)) ...$ (where ``,'' and ``$\cdot$'' both are shorthand for $App$,
 each element of the stream is the concatenation of the previous two) for query $Fib(A,B,y^v)$.
 Now let us execute the query $Take(S(S(S(Z))),y, \remph{z^v} ), Fib(A,B,y)$, Intuitively, that query  computes the prefix of length 3 in a Fibonacci word: 

\begin{center}
  $\Phi \vdash \{Take(S(S(S(Z))),y, \remph{z^v} ), Fib(A,B,y)\} \rightharpoonup_{\kappa_2, \gamma_1 \equiv [S(S(Z))/x_1, App(\remph{y_1^v}, z_1)/y, Cons(\remph{y_1^v}, \remph{r_1^v})/\remph{z^v}]} \{Take(S(S(Z)),z_1, \remph{r_1^v} ), Fib(A,B,App(\remph{y_1^v}, z_1))\}
 \rightharpoonup_{\kappa_3,\gamma_2 \equiv [A/x_2, B/y_2, A/\remph{y_1^v}, s_2/z_1]\cdot \gamma_1} 
\{Take(S(S(Z)),s_2, \remph{r_1^v} ), Fib(B, App(A,B), s_2)\}
 \rightharpoonup_{\kappa_2,\gamma_3 \equiv [S(S(Z))/x_3, App(\remph{y_3^v}, z_3)/s_2, Cons(\remph{y_3^v}, \remph{r_3^v})/\remph{r_1^v}]\cdot \gamma_2} 
\{Take(S(Z),z_3, \remph{r_3^v}), Fib(B, App(A,B), App(\remph{y_3^v}, z_3))\} 
 \rightharpoonup_{\kappa_3,\gamma_4 \equiv [B/x_4, App(A,B)/y_4, B/\remph{y_3^v}, s_4/z_3 ]\cdot \gamma_3} 
\{Take(S(Z),s_4, \remph{r_3^v}), Fib(App(A,B), App(B, (App(A, B))), s_4)\} 
 \rightharpoonup_{\kappa_2,\gamma_5 \equiv [S(Z)/x_5, App(\remph{y_5^v}, z_5)/s_4, Cons(\remph{y_5^v}, \remph{r_5^v})/\remph{r_3^v} ]\cdot \gamma_4} 
\{Take(Z,z_5, \remph{r_5^v}), Fib(App(A,B), App(B, (App(A, B))), App(\remph{y_5^v}, z_5) )\} 
 \rightharpoonup_{\kappa_1,\gamma_6 \equiv [App(x_6,y_6)/z_5, Nil/\remph{r_5^v}]\cdot \gamma_5} 
\{Fib(App(A,B), App(B, (App(A, B))), App(\remph{y_5^v}, App(x_6,y_6)) )\} $

$ \rightharpoonup_{\kappa_3,\gamma_7 \equiv [App(A,B)/x_7, App(B, (App(A, B)))/y_7, App(A,B)/\remph{y_5^v}, App(x_6,y_6)/s_7]\cdot \gamma_6}\{Fib(App(B, (App(A, B))), App(App(A,B),App(B, (App(A, B))) ), App(x_6,y_6))\}  
$
\end{center}
We can see that $[Cons(A, Cons(B, Cons(App(A,B),Nil)))/\remph{z^v}]$ is the answer relative to

\noindent $Fib(App(B, (App(A, B))), App(App(A,B),App(B, (App(A, B))) ), App(x_6,y_6))$, i.e.
the prefix of length 3 in a Fibonacci word is indeed $A, B, A\cdot B$. 
\end{example}

Note that a relative answer $\gamma$ for a query $A$ can be meaningless under Herbrand model semantics, as $\gamma A$, or indeed any of its ground instances,
may not be in the least or greatest Herbrand model of the logic program.
The following example illustrates this fact.

\begin{example}
  Consider Example~\ref{from}. Suppose the Horn formula $\kappa_1$ is replaced with the following Horn formula:
  
  $\kappa'_1 : \forall x . \forall y . False(a), From(S(x), y) \Rightarrow From(x, Cons(x,y))$\\
where $False(a)$ is a formula that does not unify with any clause in the given program. However, we can still perform the same partial LP-Unif derivation for  the query $Nth(S(Z), y, \remph{z^v}), From(S(Z), y)$, just as was shown in Example~\ref{from}. But no instance of $From(S(Z), y)$ will be contained in the least or greatest Herbrand model of that program.

  \end{example}

\subsection{Disscussion}
% about the Notion of Failure and Negation
This section  presented an experiment in applying the type-theoretic semantics formulated in the earlier section to a practical problem of establishing  a sound lazy derivation strategy for resolution.
%is quite experimental, in the sense that it 
It de-emphasizes the usual notions of refutation and  entailment in logic programming.
Based on the
type-theoretic semantics given to resolution via soundness and completeness theorems of Section \ref{pre}, every reduction path of a given query is computationally
meaningful. This generalizes the traditional declarative semantics approach to logic programming, according to which only refutation -- i.e. the reduction that lead 
to a normal form given by the empty set -- is given a model theoretic meaning. 
%In contrast, soundness result of Section~\ref{pre} do not require the query
% to be normalized to the empty set.
%, this means that we have a proof for a Horn formula 
%where the query is at the head position and the body of this Horn formula is the residuals of the reduction. This interpretation of logic programming
%gives us reason for not equating an Horn formula such as 
%$A_1,...,A_n \Rightarrow B$ with $B \Rightarrow \bot$. 
% I could not parse the below 2 sentences at all...
 
%In the light of supporting lazy computation, the
Labels we introduced in this section allowed us to annotate the intention of making an
observation, and the labelled unification was formulated to hereditarily preserve this intention. Thus, we achieved
a computational behavior that is similar to lazy functional programming languages, i.e. partial LP-Unif can
 make finite observations on the infinite data. %, and making infinite observation on 
%infinite data naturally gives rise to nonterminating behavior. %% Partial LP-Unif seems to allow
%% that, as Example \ref{from}, \ref{fib} show. 

Related work exists on supporting lazy computation in logic programming. One is by 
annotating each predicate to be inductive/coinductive \cite{Simon07}, with the intention 
of resolving the inductive predicate eagerly and memorizing the coinductive 
predicate at each step, so that one can stop the resolution whenever the current query is a variant of the previous memorized coinductive predicate.
Our approach differs in that memorization and variant detection are not needed. %use of labels seems to be a more accurate method to denote the intention of observation. 

\section{Structual Resolution}
\label{rt:s}

In this section, we represent structural resolution using the abstract reduction formalism we introduced in Section \ref{pre}.
We first define structural resolution as \textit{LP-Struct} reduction, thereby also defining \textit{LP-TM} reduction, which replaces the unification in LP-Unif by term-matching. We then identify
two conditions under which LP-Unif and LP-Struct are operationally equivalent. These two conditions are the termination of LP-TM and the non-overlapping requirement for Horn clauses. 

%% Finally, we show how realizability transformation can be used to guarantee the termination of LP-TM and the non-overlapping of logic program. Since realizability transformation
%% preserves the meaning of the logic program, we establish the equivalence of LP-Struct and LP-Unif. 

%prove its completeness relative

\begin{definition}
\

  \begin{itemize}
  \item \textbf{Term-matching (LP-TM) reduction:}

$\Phi \vdash \{A_1,..., A_i, ..., A_n\} \to_{\kappa} \{A_1,..., \sigma B_1,..., \sigma B_m, ..., A_n\}$, if there exists $\kappa : \forall \underline{x} . B_1,..., B_n \Rightarrow C \in \Phi$ such that $\sigma C \equiv A_i$.

\item \textbf{Substitutional reduction:} 

$\Phi \vdash \{A_1,..., A_i, ..., A_n\} \hookrightarrow_{\kappa, \gamma \cdot \gamma'} \{\gamma A_1,..., \gamma A_i, ..., \gamma A_n\}$ for any substitution $ \gamma'$, if there exists
  
  $\kappa : \forall \underline{x} . B_1,..., B_n \Rightarrow C \in \Phi$ such that $C \sim_{\gamma} A_i$.

\end{itemize}

\end{definition}

When combining LP-TM reduction with substitutional reductions, we sometimes write $\to_{\kappa, \gamma}$, where the second subscript is used to store the substitution $\gamma$.
The second subscript in 
the substitutional reduction is intended as a state (similar to LP-Unif), it will be updated along with reduction. 

% LP-TM seems to be a foreign notion for LP, but it is used in \textit{Context Reduction} \cite{Jones97} in type class instance resolution.
Given a program $\Phi$ and
a set of queries $\{B_1, \ldots, B_n\}$, LP-TM uses only term-matching reduction to reduce $\{B_1, \ldots, B_n\}$: 

\begin{definition}[LP-TM]
\
\noindent Given a logic program $\Phi$,  LP-TM is given by an abstract reduction system $(\Phi, \to)$. % namely, given
%a query $B$, LP-TM uses only term-matching reduction to reduce $\{B\}$.
\end{definition}

LP-TM is also sound w.r.t. the type system
of Definition \ref{proofsystem}, which implies that we can obtain a proof for each reductoin of the query. 

\begin{theorem}[Soundness of LP-TM]
\label{sound:tm}
    If $\Phi \vdash \{A\} \to^* \{B_1,..., B_n\}$ , then $\Phi \vdash e : \forall \underline{x} .B_1,..., B_n \Rightarrow A$ in $\mathbf{C}$.
\end{theorem} 

%% \begin{proof}
%%   Similar to the proof of Lemma \ref{sound}.
%% \end{proof} 

Comparing the soundness lemma and Theorem \ref{sound:tm}, we see that for LP-TM, there is no need to accumulate substitutions, this is due to the use of term-matching instead of unification for the LP-TM reduction. %% The following example shows that the LP-TM is incomplete with respect to the type system. 

%% \begin{example}
%% Consider the following program $\Phi$.

%% \begin{center}
%%   $\kappa_1 :\ \Rightarrow Q({C})$

%%   $\kappa_2 :\ \forall x . \forall y. Q(x) \Rightarrow P(y)$
%% \end{center}

%% \noindent For query $P({C})$, we have $\Phi \vdash \{P({C})\} \to_{\kappa_2} \{Q(x)\} \not \to$. However, there exist a proof $\Phi \vdash \kappa_2 \ \kappa_1 :\ \Rightarrow P({C})$ in 
%% \textbf{C}.
%% \end{example}

  We use $\to^\mu$ to denote a reduction path to a $\to$-normal form. If the
  $\to$-normal form does not exist, then $\to^\mu$ denotes an infinite reduction path. We write $\hookrightarrow^1$ to denote at most one step of $\hookrightarrow$.

%i.e. every $\to$-reduction path is infinite
%% If we know that $\to$ is strongly normalizing, then we use $\to^\mu$ to denote a reduction path to a $\to$-normal form.

%For a LP-TM terminating program is defined using only the term-matching reduction, for a prductive program, it is still possible that for a query to be diverged, due to the present
% of unification(See Example \ref{ex:str}).   

We can now formally define structural resolution within our formal framework. Given a program $\Phi$ and
a set of queries $\{B_1, \ldots, B_n\}$, LP-Struct  first uses term-matching reduction to reduce  $\{B_1, \ldots, B_n\}$ to a normal form, then performs one step substitutional reduction, and then repeats this process.
%to reduce $\{B_1, \ldots, B_n\}$

\begin{definition}[Structural Resolution (LP-Struct)]
\
\noindent  Given a logic program $\Phi$,  LP-Struct is given by an abstract reduction system $(\Phi, \to^{\mu} \cdot \hookrightarrow^1)$.
%LP by structural resolution is defined as $(\Phi, \to^{\mu} \cdot \hookrightarrow^1)$, namely, given a query $B$, LP-Struct

\end{definition}

If a finite term-matching reduction path does not exist, then $\to^{\mu} \cdot \hookrightarrow^1$ denotes an infinite path. 
When we write $\Phi \vdash \{\underline{A}\} (\to^{\mu} \cdot \hookrightarrow^1)^* \{\underline{C}\}$, it means a nontrivial finite path will be of the shape $\Phi \vdash \{\underline{A}\} \to^{\mu} \cdot \hookrightarrow \cdot ... \cdot \to^{\mu} \cdot \hookrightarrow \cdot \to^\mu \{\underline{C}\}$. 

Now let us recall the execution trace of the query $Stream(x)$ in Example \ref{ex:zstream} using LP-Struct: 

  \begin{center}
    $\Phi \vdash \{{Stream}(x)\} \hookrightarrow_{\kappa_1, [{Cons}(0, y_1)/x]} \{{Stream}({Cons}(0,y_1))\} \to_{\kappa_1}
    \{{Stream}(y_1)\}  \hookrightarrow_{\kappa_1, [{Cons}(0, y_2)/y_1, {Cons}(0, {Cons}(0, y_2))/x ]} \{{Stream}({Cons}(0, y_2))\} \to_{\kappa_1} 
    \{{Stream}(y_2)\} \hookrightarrow_{\kappa_1,  [{Cons}(0, y_3)/y_2, {Cons}(0,{Cons}(0, y_3) )/y_1, {Cons}(0, {Cons}(0, {Cons}(0, y_3)))/x]} \{{Stream}({Cons}(0, y_3))\} \to_{\kappa_1} \{{Stream}(y_3)\} \hookrightarrow ... $
  \end{center}

\subsection{LP-Struct and LP-Unif}

LP-Struct exibits different execution behavior compared to LP-Unif. In general, they are not equivalent.
Consider the program and the finite LP-Unif derivation of Example~\ref{ex:conn}. LP-Unif has
 a finite successful derivation for the query ${Connect}(x, y)$, but we have the following non-terminating reduction by LP-Struct:  

\begin{center} 
  $\Phi \vdash \{{Connect}(x, y)\} \to_{\kappa_1} \{{Connect}(x, y_1), {Connect}(y_1, y)\} $

$\to_{\kappa_1} \{{Connect}(x, y_2), {Connect}(y_2, y_1), {Connect}(y_1, y)\} \to_{\kappa_1} ... $ 
\end{center}
 The diverging behavior above is due to the divergence of LP-TM reduction.

%%  \knote{In the intro, I define productivity as conjunction of 2 conditions: LP-TM termination and nontermination of the overall LP-Struct. So, i am changing the title below, ok?
%% we need to chase any other occurences...
%%  }
 
\begin{definition}[LP-TM Termination]
  We say a program $\Phi$ is LP-TM terminating iff it admits no infinite $\to$-reduction.
\end{definition}

LP-TM termination is important for LP-Struct in two aspects: 1. It is
one of the conditions that ensure the operational equivalence of LP-Struct and LP-Unif. 2. The 
finiteness of LP-TM reduction is used in defining the observational productivity in logic programming.
%in LP-Struct contributes to the finite observation if the whole
%LP-Struct is nonterminating. 

The following example shows that LP-TM termination  alone
is not sufficient to establish that LP-Unif and LP-Struct are operationally equivalent. 
\begin{example}
\label{ex:overlap}
Consider the following logic program (we use $K$ to denote a constant):

  \begin{center}

    \noindent $\kappa_1 :\ \Rightarrow P({K})$ 

   \noindent  $\kappa_2 :  \forall x . Q(x) \Rightarrow P(x)$
  \end{center}

\noindent  The program is LP-TM terminating. For query $P(x)$, we have $\Phi \vdash \{P(x)\} \leadsto_{\kappa_1, [{K}/x]} \emptyset$ with LP-Unif, but there is only one reduction path $\Phi \vdash \{P(x)\} \to_{\kappa_2} \{Q(x)\} \not \hookrightarrow$ for LP-Struct.
\end{example}

Thus the termination of LP-TM is insufficient for establishing the relation between LP-Struct and LP-Unif. In Example \ref{ex:overlap}, the problem is caused by
the overlapping heads $P(K)$ and $P(x)$. Motivated by the non-overlapping condition
for rewrite rules in term rewriting systems (\cite{Baader:1998}, \cite{bezem2003term}), we introduce the following definition.

\begin{definition}[Non-overlapping Condition]
  Axioms $\Phi$ are non-overlapping if for any $\kappa_i : \forall \underline{x}. \underline{B} \Rightarrow C, \kappa_j : \forall \underline{x}. \underline{D} \Rightarrow E \in \Phi$, there are no substitutions $\sigma, \delta$ such that $\sigma C \equiv \delta E$.
\end{definition}

The following lemma shows that an LP-TM step can be viewed as an LP-Unif step without affecting 
the accumulated substiution. 
%% But the termination requirement can be weaken by only requiring
%% termination of the $\to$-reduction, i.e. by requiring the termination of LP-TM. 

\begin{lemma}
  \label{tm-to-unif}
If $\Phi \vdash \{D_1,..., D_i,..., D_n\} \to_{\kappa, \gamma} \{D_1,.., \sigma E_1,..., \sigma E_m,..., D_n\}$, with $\kappa : \forall \underline{x}. \underline{E} \Rightarrow C \in \Phi$ and 
$\sigma C \equiv D_i$ for any $\gamma$, then $\Phi \vdash \{D_1,..., D_i, ..., D_n\} \leadsto_{\kappa, \gamma} \{D_1,.., \sigma E_1,..., \sigma E_m,..., D_n\}$. 
\end{lemma}
\begin{proof}
  Since for $\Phi \vdash \{D_1,..., D_i,..., D_n\} \to_{\kappa, \gamma} \{D_1,.., \sigma E_1,..., \sigma E_m,..., D_n\}$, with $\kappa : \forall \underline{x}. \underline{E} \Rightarrow C \in \Phi$ and $\sigma C \equiv D_i$, we have $\Phi \vdash \{D_1,..., D_i, ..., D_n\} \leadsto_{\kappa, \sigma \cdot \gamma} \{\sigma D_1,.., \sigma E_1,..., \sigma E_m,..., \sigma D_n\}$. But $\mathrm{dom}(\sigma) \in \mathrm{FV}(C)$, thus we have $\Phi \vdash \{D_1,..., D_i, ..., D_n\} \leadsto_{\kappa, \gamma} \{D_1,.., \sigma E_1,..., \sigma E_m,..., D_n\}$. 
\end{proof}

The following lemma shows that one LP-Struct step corresponds to several LP-Unif steps, given 
the non-overlapping requirement. 

\begin{lemma}\label{struct-to-unif}
Suppose $\Phi$ is non-overlapping and $\{A_1,..., A_n\}$ are $\to$-normal. If $\Phi \vdash \{A_1,..., A_n\}  (\hookrightarrow_{\kappa, \gamma} \cdot \to^\mu_\gamma) \{C_1, ..., C_m\}$, then $\Phi \vdash \{A_1, ..., A_n\} \leadsto^*_\gamma \{C_1, ..., C_m\}$. 
\end{lemma}
\begin{proof}
  Given $\Phi \vdash \{A_1,..., A_n\} (\hookrightarrow_{\kappa, \gamma} \cdot \to^\mu_\gamma) \{C_1, ..., C_m\}$, we know the actual reduction path can be rearranged to the form $\Phi \vdash \{A_1,..., A_n\}\hookrightarrow_{\kappa, \gamma} \{\gamma A_1,..., \gamma A_n\} \to_{\kappa, \gamma} \{\gamma A_1, ..., \gamma B_1, ..., \gamma B_n,..., \gamma A_n\} \to^\mu_{\gamma} \{C_1, ..., C_m\}$, where $\gamma A_i \equiv \gamma C$ with $\kappa : \forall \underline{x}.\underline{B} \Rightarrow C \in \Phi$ and $A_i \equiv \sigma B$. Note that $\gamma$ is unchanged along the term-matching reduction. We have rearranged the $\to$-step following right after $\hookrightarrow$ using $\kappa$ due to the property of LP-TM. Note that to LP-TM reduce $A_i$ we can only use $\kappa$, since otherwise it would mean with $\kappa' : \forall \underline{x}. \underline{D} \Rightarrow B \in \Phi$. This implies $\gamma C \equiv \gamma \sigma B$, contradicting
the non-overlapping restriction. Thus we have $\Phi \vdash \{A_1,..., A_n\}\leadsto_{\kappa,  \gamma} \{\gamma A_1, ..., \gamma B_1, ..., \gamma B_n,..., \gamma A_n\}$. By Lemma \ref{tm-to-unif},
we have $\Phi \vdash \{A_1,..., A_n\}\leadsto_{\kappa,  \gamma} \{\gamma A_1, ..., \gamma B_1, ..., \gamma B_n,..., \gamma A_n\}$ $\leadsto^*_\gamma \{C_1, ..., C_m\}$. 
\end{proof}

Lemmas \ref{lemm:norm}, \ref{lem:str} and Theorem \ref{ortho:equiv}  show that for a non-overlapping program, LP-Unif is equivalent to LP-Struct for terminating reductions. 

\begin{lemma}
  \label{lemm:norm}
  Given $\Phi$ is non-overlapping, if $\Phi \vdash \{A_1,..., A_n\} (\to^\mu \cdot \hookrightarrow^1)^*_\gamma  \{C_1, ..., C_m\}$ with $\{C_1, ..., C_m\}$ in $\to^\mu \cdot \hookrightarrow^1$-normal form, then $\Phi \vdash \{A_1, ..., A_n\} \leadsto^*_\gamma \{C_1, ..., C_m\}$ with $\{C_1, ..., C_m\}$ in $\leadsto$-normal form. 
\end{lemma}
\begin{proof}
  Since $\Phi \vdash \{A_1, ..., A_n\}  (\to^\mu \cdot \hookrightarrow^1)^*_\gamma \{C_1, ..., C_m\}$, this means the reduction path must be of the form $\Phi \vdash \{A_1,..., A_n\} \to^\mu \cdot \hookrightarrow^1 \cdot \to^\mu \cdot \hookrightarrow^1 ... \to^\mu \cdot \hookrightarrow^1 \cdot \to^\mu  \{C_1, ..., C_m\}$. Thus $\Phi \vdash \{A_1,..., A_n\} \to^\mu \cdot (\hookrightarrow^1 \cdot \to^\mu) \cdot (\hookrightarrow^1 ... \to^\mu) \cdot (\hookrightarrow^1 \cdot  \to^\mu)  \{C_1, ..., C_m\}$. By Lemma \ref{tm-to-unif} and Lemma \ref{struct-to-unif}, we have $\Phi \vdash \{A_1, ..., A_n\} \leadsto^*_\gamma \{C_1, ..., C_m\}$ with $\{C_1, ..., C_m\}$ in $\leadsto$-normal form.
\end{proof}

\begin{lemma}
  \label{lem:str}
  Given $\Phi$ is a non-overlapping, if $\Phi \vdash \{A_1, ..., A_n\} \leadsto^*_\gamma \{C_1, ..., C_m\}$ with $\{C_1, ..., C_m\}$ in $\leadsto$-normal form , then $\Phi \vdash \{A_1, ..., A_n\} (\to^\mu \cdot \hookrightarrow^1)^*_\gamma \{C_1, ..., C_m\}$ with $\{C_1, ..., C_m\}$ in $\to^\mu \cdot \hookrightarrow^1$-normal form. 
\end{lemma}
\begin{proof}
  By induction on the length of $\leadsto^*_\gamma$. 

\begin{itemize}
\item Base Case. $\Phi \vdash \{A_1,...,A_i, ..., A_n\} \leadsto_{\kappa, \gamma} \{\gamma A_1,...,\gamma B_1, ..., \gamma B_m ..., \gamma A_n \}$ with $\kappa : \forall \underline{x}. \ \underline{B}\Rightarrow C \in \Phi$, $C \sim_\gamma A_i$ and $\{\gamma A_1,...,\gamma B_1, ..., \gamma B_m ..., \gamma A_n \}$ in $\leadsto$-normal form . We have $\Phi \vdash \{A_1,...,A_i, ..., A_n\} \hookrightarrow_{\kappa, \gamma} \{\gamma A_1,...,\gamma A_i, ..., \gamma A_n \} \to_\kappa \{\gamma A_1,...,\gamma B_1, ..., \gamma B_m ..., \gamma A_n\}$ with $\{\gamma A_1,...,\gamma B_1, ..., \gamma B_m ..., \gamma A_n\}$ in $\to^\mu\cdot \hookrightarrow$-normal form. Note that there can not be another
$\kappa' : \forall \underline{x}. \underline{B} \Rightarrow C' \in \Phi$ such that $\sigma C' \equiv A_i$, since this would means $\gamma C \equiv \gamma A_i \equiv \gamma \sigma C'$, violating the non-overlapping requirement.

\item Step Case. $\Phi \vdash \{A_1, ..., A_i, ..., A_n\} \leadsto_{\kappa, \gamma} \{\gamma A_1,..., \gamma B_1, ..., \gamma B_l, ..., \gamma A_n\} \leadsto^*_{\gamma'} \{C_1, ..., C_m\}$ with $\kappa : \forall \underline{x}. B_1,..., B_l \Rightarrow C \in \Phi$ and $C \sim_\gamma A_i$.  We have $\Phi \vdash \{A_1, ..., A_i, ..., A_n\} \hookrightarrow_{\kappa, \gamma} \{\gamma A_1, ..., \gamma A_i, ..., \gamma A_n\}\to \{\gamma A_1,..., \gamma B_1, ..., \gamma B_m, ..., \gamma A_n\}$. By the non-overlapping requirement, there can not be another $\kappa' : \forall \underline{x}. \underline{D} \Rightarrow C' \in \Phi$ such that $\sigma C' \equiv A_i$. By IH, we know $\Phi \vdash \{\gamma A_1,..., \gamma B_1, ..., \gamma B_m, ..., \gamma A_n\} (\to^\mu \cdot \hookrightarrow)_{\gamma'}^* \{C_1, ..., C_m\}$. Thus we conclude that $\Phi \vdash \{A_1, ..., A_i, ..., A_n\} (\hookrightarrow \cdot \to)^*_{\gamma'} \{C_1, ..., C_m\}$. 
\end{itemize}
\end{proof}

\begin{theorem}
  \label{ortho:equiv}
  Suppose $\Phi$ is non-overlapping. $\Phi \vdash \{A_1, ..., A_n\} \leadsto^*_\gamma \{C_1, ..., C_m\}$ with $\{C_1, ..., C_m\}$ in $\leadsto$-normal form iff $\Phi \vdash \{A_1, ..., A_n\} (\to^\mu \cdot \hookrightarrow^1)^*_\gamma \{C_1, ..., C_m\}$ with $\{C_1, ..., C_m\}$ in $\to^\mu \cdot \hookrightarrow^1$-normal form.   
\end{theorem}

The theorem above implies that for terminating and non-overlapping programs, LP-Unif is equivalent
to LP-Struct. But the termination requirement can be relaxed by only requiring
the termination of the $\to$-reduction, i.e. by requiring termination of LP-TM.

%% \begin{lemma}
%% Suppose $\Phi$ is a non-overlapping and LP-TM terminating program. If $\Phi \vdash \{A_1,..., A_n\} \leadsto \{B_1,..., B_m\}$, then $\Phi \vdash \{A_1,..., A_n\} (\to^\mu \cdot \hookrightarrow^1)^* \{C_1,..., C_l\}$ and $\Phi \vdash \{B_1,..., B_m\} \to^* \{C_1,..., C_l\}$.
%% \end{lemma}
%% \begin{proof}
%% \end{proof}
%% \begin{lemma}\label{I}
%% Suppose $\Phi$ is a non-overlapping and LP-TM terminating program.  \end{lemma}
%% \begin{proof}
%% \end{proof}
\begin{lemma}
  \label{lem:subst}
 If $\Phi \vdash \{A\} \to_{\kappa, \gamma} \{B_1,..., B_m\}$ and
 $\mathrm{dom}(\sigma) \cap ((\bigcup_i \mathrm{FV}(B_i)) - \mathrm{FV}(A)) = \emptyset$, then $\Phi \vdash \{\sigma A\} \to_{\kappa, \gamma} \{\sigma B_1,..., \sigma B_m\}$. 
\end{lemma}
Note that the above lemma shows that the reduction $\to$ is closed under substitution only under the condition
that $\mathrm{dom}(\sigma) \cap ((\bigcup_i \mathrm{FV}(B_i)) - \mathrm{FV}(A)) = \emptyset$, i.e. the domain of the substitution must not contain any variable that are in $B_i$ but not in $A$ for any $i$, otherwise it will not be the case. If $\Phi \vdash \{A\} \to^\mu \{B_1,..., B_m\}$, 
we write $[A]$ to mean the normal form of $A$, i.e. $B_1,..., B_m$. 

\begin{theorem}[Equivalence of LP-Struct and LP-Unif]\label{prod-non-overlap}
  Suppose $\Phi$ is non-overlapping and LP-TM terminating. 
  \begin{enumerate}
  \item If $\Phi \vdash \{A_1,..., A_n\} \leadsto \{B_1,..., B_m\}$, then $\Phi \vdash \{A_1,..., A_n\} (\to^\mu \cdot \hookrightarrow^1)^* \{C_1,..., C_l\}$ and $\Phi \vdash \{B_1,..., B_m\} \to^* \{C_1,..., C_l\}$.
   \item If $\Phi \vdash \{A_1,..., A_n\} (\to^\mu \cdot \hookrightarrow^1)^* \{B_1,..., B_m\}$, then $\Phi \vdash \{A_1,..., A_n\} \leadsto^* \{B_1,..., B_m\}$.  
  \end{enumerate}
\end{theorem}

\begin{proof}
  \begin{enumerate}
  \item   Suppose $\Phi \vdash \{A_1,..., A_n\} \leadsto_{\kappa, \gamma} \{\gamma A_1,...,\gamma E_1, ..., \gamma E_l,..., \gamma A_n\}$, with $\kappa : \underline{E} \Rightarrow D \in \Phi$ and $D \sim_\gamma A_i$. Suppose $\gamma D \equiv A_i$, we have $\Phi \vdash \{A_1,..., A_n\} \to_{\kappa, \gamma} \{\gamma A_1,...,\gamma E_1, ..., \gamma E_q,..., \gamma A_n\} \to^\mu_\gamma \{C_1,..., C_l\}$. Suppose $\gamma D \not \equiv  A_i$. In this case, we have $\Phi \vdash \{A_1,..., A_n\}  \to^\mu \{[A_1], ..., A_i, ..., [A_n]\} \hookrightarrow_{\kappa, \gamma} \{[\gamma A_1], ..., \gamma A_i, ..., [\gamma A_n]\}  \to_{\kappa, \gamma} \{[\gamma A_1],...,\gamma E_1, ..., \gamma E_q,..., [\gamma A_n]\} \to^\mu_\gamma \{C_1,..., C_l\}$. By Lemma \ref{lem:subst}, we know that
    $\Phi \vdash {\gamma A_1} \to^\mu \{[\gamma A_1]\}$, ..., $\Phi \vdash {\gamma A_n} \to^\mu \{[\gamma A_n]\}$. Thus $\Phi \vdash \{\gamma A_1,...,\gamma E_1, ..., \gamma E_q,..., \gamma A_n\} \to^\mu \{C_1,..., C_l\}$.

    \item We just need to show that if $\{A_1,..., A_n\}$ are $\to$-normal and $\Phi \vdash \{A_1,..., A_n\} \hookrightarrow_{\kappa,\gamma} \{\gamma A_1,..., \gamma A_n\} \to^\mu_\gamma \{B_1,..., B_m\}$, then $\Phi \vdash \{A_1,..., A_n\} \leadsto^*_\gamma \{B_1,..., B_m\}$. Suppose $\Phi \vdash \{A_1,..., A_n\} \hookrightarrow_{\kappa,\gamma} \{\gamma A_1,..., \gamma A_n\} \to^\mu_\gamma \{B_1,..., B_m\}$, we have $\Phi \vdash \{A_1,..., A_n\} \hookrightarrow_{\kappa,\gamma} \{\gamma A_1,..., \gamma A_n\} \to_\kappa \{\gamma A_1,..., \gamma C_1,..., \gamma C_l,..., \gamma A_n\}\to^\mu_\gamma \{B_1,..., B_m\}$ with $\kappa : \underline{C} \Rightarrow D \in \Phi$ and $D \sim_\gamma A_i$. Thus we have
      
\noindent $\Phi \vdash \{A_1,..., A_n\} \leadsto_{\kappa,\gamma}  \{\gamma A_1,..., \gamma C_1,..., \gamma C_l,..., \gamma A_n\}$. By Lemma \ref{tm-to-unif}, we have $\Phi \vdash \{A_1,..., A_n\} \leadsto_{\kappa,\gamma}  \{\gamma A_1,..., \gamma C_1,..., \gamma C_l, ...., \gamma A_n\} \leadsto^*_\gamma \{B_1,..., B_m\}$.

  \end{enumerate}
\end{proof}
Note that the above theorem does not rely on the whole program termination, therefore it establishes equivalence of LP-Unif and LP-Struct even for nonterminating programs 
such as the Stream example, as long
as they are LP-TM terminating and non-overlapping. 
This result has not been described in previous work.

\subsection{Discussion}

Structural resolution was first introduced in Komendantskaya and Power's work 
(\cite{KomendantskayaP11}, \cite{komendantskaya2014}) under the name of coalgebraic logic programming. It was further developed into a resolution method  
 based on resolution trees (called \textit{rewriting trees})  generated by term-matching~(\cite{JKK15}, \cite{KJ15}).
The formulation of LP-Struct in this paper is based on the abstract
reduction system framework, instead of the tree formalism in previous work. As a consequence,
for overlapping logic programs, the reduction-based LP-Struct behaves differently compared to the tree-based formalism (see e.g. Example \ref{ex:overlap}).
The novelty of our development in this section is the articulation of the two conditions that ensure the operational equivalence of LP-Struct and LP-Unif (Theorem \ref{prod-non-overlap}). 

%% realizability transformation and the demonstration that the transformation
%% preserves the proof-theoretic meaning of the logic program. As a consequence, we obtain a
%% useful technique to ensure the completeness of the structual resolution, which previous is 
%% thought to be complete for the LP-TM terminating program. We also establish the operational
%% equivalence of SLD resolution and structural resolution (for both terminating and nonterminating query) via the transformation 
%% (Corollary \ref{equiv}), while previous work focus on the soundness and completeness with respect to the Herbrand model.  

\section{Functionalisation of LP-TM}
\label{s:func}
One of the features of LP-Struct is that it refines SLD-resolution by
  a combination of term-matching and unification. LP-TM itself is used in
  the type class context reduction \cite{Jones97}. The termination behavior 
of LP-TM is of practical interest. For example, termination for the type class inference is essential to achieve decidability of the type inference in languages such as Haskell (\cite{Lammel:2005} Section 5). 

Of course, the termination of LP-TM also implies observational productivity 
in the context of LP-Struct. As explained in Introduction and Section \ref{rt:s}, 
LP-TM termination is not only essential to ensure the equivalence of
LP-Struct and LP-Unif, but also is important to allow viewing the LP-TM reductions within 
 the nonterminating LP-Struct reduction as finite observations. 

On the other hand, termination and nontermination detection are well-studied in the context
of term rewriting. In this section we show a method that reuses the techniques developed in term rewriting
to detect termination of LP-TM. We first define a process called \textit{functionalisation} that transforms 
a set of Horn clauses into a set of rewrite rules, where the execution of a query is seen as a
process of rewriting the query to its proof. As a result, termination and nontermination detection techniques from term rewriting can be applied to LP-TM, assuming the logic program
contains no existential variables.

In this section we work only with the Horn formulas without \textit{existential variables}, i.e.
for any Horn formula $\forall \underline{x}. A_1,..., A_n \Rightarrow B$, we have $\bigcup_i \mathrm{FV}(A_i) \subseteq \mathrm{FV}(B)$.
The restriction that Horn clauses should not contain existential variables comes directly from a similar requirement imposed in term-rewriting.

Since the idea of functionalisation is to view LP-TM resolution for a query as a rewriting process to its proof evidence, the rewriting is defined
on \textit{mixed terms}, i.e. a mixture of atomic formulas and proof evidence. 

\begin{definition}[Mixed Terms]

\[\begin{array}{llll}
  \text{Mixed term} & q \ & ::= & \ A \ |\ \kappa \ | \ q \ q' \\
  \text{Mixed term context} & \mathcal{C} \ & ::= & \ \bullet \ | \ \mathcal{C}\ q \ | \ q\ \mathcal{C}\\
\end{array}
\]
\end{definition}
Note that $\mathcal{C}$ can be ground. Let $\mathcal{C}[q_1, ..., q_n]$ mean replacing all the $\bullet$ in $\mathcal{C}$ from left to right by $q_1, ..., q_n$.

\begin{definition}[Functionalisation]
  \label{label}
We can construct a set of rewrite rule $K(\Phi)$ from a set of axioms $\Phi$ as follows.  For each $\kappa :\forall \underline{x}. A_1,..., A_n \Rightarrow B \in \Phi$, we define a rewrite rule $B \to \kappa\ A_1...\ A_n \in K(\Phi)$ on mixed terms. We call $\kappa$ an axiom symbol.
\end{definition}
Note that the evidence constant $\kappa$ for $A_1,..., A_n \Rightarrow B$ becomes a
mixed term function symbol of arity $n$, with $A_1,..., A_n$ as its arguments, which is denoted by $\kappa\ A_1...\ A_n$. 

\begin{definition}
  We define a relation $\mathcal{C}\to \mathcal{C}'$ to mean that $\mathcal{C}'$ can be obtained from $\mathcal{C}$ by replacing a $\bullet$ in $\mathcal{C}$ by some $\mathcal{C}_1$, where $\mathcal{C}_1 \not \equiv \bullet$. 
We also write the reflexive and transitive closure of this relation as $\to^*$.
\end{definition}

The following lemmas show that each LP-TM step corresponds exactly to a rewrite step after functionalisation. As a consequence, it is possible to determine the termination behavior of
LP-TM by analyzing the corresponding term rewriting system.  
\begin{lemma}
\label{short}
$\Phi \vdash \{A_1 , ..., A_n\} \to \{A_1, ..., \sigma B_1, ..., \sigma B_m, ..., A_n\}$, where $\kappa: \forall \underline{x}. B_1 , ..., B_m \Rightarrow B \in \Phi$ and $\sigma B \equiv A_i$ iff $\mathcal{C}[A_1, ..., A_i, ..., A_n] \to \mathcal{C}'[A_1,...,\sigma B_1, ..., \sigma B_m, ..., A_n]$, where $\mathcal{C}, \mathcal{C}'$ do not contain any atomic formulas and $\mathcal{C} \to \mathcal{C}'$.
\end{lemma}
\begin{proof}

  By Definition \ref{label},  $\kappa: \forall \underline{x}.  B_1 , ..., B_m \Rightarrow A_i \in \Phi$ implies $A_i \to \kappa\ B_1\ ...\ B_m \in K(\Phi)$, and vice versa. So $\mathcal{C}'$ can be obtained by replacing the $i$th $\bullet$ in $\mathcal{C}$ by $\kappa\ \bullet_1\ ...\ \bullet_m$.
\end{proof}

\begin{lemma}
\label{long}
$\Phi \vdash \{A_1,..., A_n\} \to^* \{C_1, ..., C_l\}$ iff $\mathcal{C}[A_1, ..., A_n] \to^* \mathcal{C}'[C_1, ..., C_l]$, where $\mathcal{C} \to^* \mathcal{C}'$ and $\mathcal{C}, \mathcal{C}'$ do not contain any atomic formulas.
\end{lemma}
\begin{proof}
  We prove both direction together. By induction on the length of $\to^*$.

\noindent Base Case. By Lemma \ref{short}.

\noindent Step Case.

\noindent \textit{Left to Right}: 

\noindent Suppose $\Phi \vdash \{A_1,..., A_n\} \to \{A_1, ...,\sigma B_1, ..., \sigma B_m, ..., A_n\} \to^* \{C_1, ..., C_l\}$, with $\kappa : \forall \underline{x}. B_1, ..., B_m \Rightarrow B \in \Phi$ and $\sigma B \equiv A_i$. Then we know $\mathcal{C}[A_1,..., A_n] \to \mathcal{C}''[A_1, ..., \sigma B_1, ..., \sigma B_m, ..., A_n]$. Also, $\mathcal{C} \to \mathcal{C}''$, where $\mathcal{C}''$ can be obtained from $\mathcal{C}$ by replacing its $i$th $\bullet$ by $\kappa \ \bullet_1 ... \ \bullet_m$. By IH, 
$\mathcal{C}''[A_1, ..., \sigma B_1, ..., \sigma B_m, ..., A_n] \to^* \mathcal{C}'[C_1, ..., C_l]$ with $\mathcal{C}'' \to^* \mathcal{C}'$. So $\mathcal{C}[A_1,..., A_n]  \to^* \mathcal{C}'[C_1, ..., C_l]$ with $\mathcal{C} \to^* \mathcal{C}'$. 

\noindent \textit{Right to Left}: 

\noindent Suppose $\mathcal{C}[A_1,..., A_n] \to \mathcal{C}''[A_1, ..., \sigma B_1, ..., \sigma B_m, ..., A_n] \to^* \mathcal{C}'[C_1, ..., C_l]$ with $\mathcal{C}\to \mathcal{C}'' \to^* \mathcal{C}'$, where $B \to \kappa \ B_1\ ... \ B_m$ and $\sigma B \equiv A_i$. So $\forall \underline{x}. B_1, ..., B_m \Rightarrow B \in \Phi$ and $\sigma B \equiv A_i$. Thus $\Phi \vdash \{A_1,..., A_n\} \to \{A_1, ...,\sigma B_1, ..., \sigma B_m, ..., A_n\}$. By IH, $\Phi \vdash \{A_1, ...,\sigma B_1, ..., \sigma B_m, ..., A_n\} \to^* \{C_1, ..., C_l\}$. Thus, we have $\Phi \vdash \{A_1,..., A_n\} \to^* \{C_1, ..., C_l\}$.
\end{proof}

\begin{theorem}
 $\Phi \vdash \{A\} \to^* \emptyset$ iff $A \to^* e$, and $e$ is a ground
  evidence. As a consequence, the query $A$ is LP-TM (non)terminating iff $A$
  is (non)terminating for $K(\Phi)$.
\end{theorem}
\begin{proof}
By Lemma \ref{long}. 
\end{proof}

In practice, functionalisation can also be used to implement
LP-TM, especially if computing the proof evidence is the only goal. For example, this is the case for 
type class inference \cite{flops16}. 

Now we demonstrate how to apply
a convenient termination technique in term rewriting called \textit{dependency pair} 
method \cite{Arts2000} to analyze the termination behavior of LP-TM. 
%The next definitions and a theorem follow~\cite{Arts2000}. 
\begin{definition}
  We define the dependency pairs generated from $K(\Phi)$ to be $E(K(\Phi)) = \{ B \to A_i\ |\ B \to \kappa \ A_1 \ ...\ A_n \in K(\Phi)\}$.
\end{definition}
%The structure of $K(\Phi)$ gives us a very simple definition of $K(\Phi)$-chain.
\begin{definition}
A (potentially infinite) sequence of pairs $q_1 \to q'_1, q_2 \to q'_2, ... $ in $E(K(\Phi))$ is
a $K(\Phi)$-chain iff there is a substitution $\sigma$ with $\sigma q_i' \equiv \sigma q_{i+1}$ for all $i$.   
\end{definition}

The above definition of $K(\Phi)$-chain is using the condition $\sigma q'_i \equiv \sigma q_{i+1}$ instead of $\sigma q'_i \to^* \sigma q_{i+1}$ in term rewriting \cite{Arts2000}. Since the dependency pairs generated from a logic program will always be in the form of $A \to B$, where $A, B$ are atomic formulas, rewriting under the predicate is not possible.
This greatly simplifies the termination detection for LP-TM. 

\begin{theorem}[Arts-Giesl \cite{Arts2000}]
\label{arts}
  $K(\Phi)$ is terminating iff no infinite $K(\Phi)$-chain exist.
\end{theorem}

Theorem \ref{arts} allows us to detect the termination of $K(A)$ by looking at the possible $K(\Phi)$-chain.

\begin{example}
Consider the following program $\Phi$:

\begin{center}
\noindent $\kappa_1 : \ \Rightarrow P({Int})$  

\noindent $\kappa_2 : \forall x . P(x), P({List}(x)) \Rightarrow P({List}(x))$
\end{center} 
The dependency pairs of $\Phi$ are $P({List}(x)) \to P({List}(x))$ and $P({List}(x)) \to P(x)$.
We can see $P({List}(x)) \to P({List}(x))$ can form an infinite $E(K(\Phi))$-chain, thus $K(\Phi)$ is not terminating. So $\Phi$ is not LP-TM terminating. 

\end{example}
 
\section{Realizability Transformation and LP-Struct}
\label{s:real}
Functionalisation provides a way to detect LP-TM termination for 
LP-Struct. But sometimes there are logic programs that are not LP-TM terminating but are still 
meaningful from the LP-Unif perspective (cf. Example \ref{ex:conn}). For these programs, we still
want to be able to use LP-Struct. To solve this problem, 
we define a meaning preserving \textit{realizability transformation} that transforms
any logic program into LP-TM terminating one. 

Realizability \cite{KleeneSC:1952}(\S 82) is a technique that uses a number representing the proof of a number-theoretic formula. The transformation described here is similar in the sense that we use a first-order term to represent the proof of a Horn formula. 
More specifically, we use a first-order term as an extra argument for Horn formula to represent a proof of that formula. 

Lemma \ref{fst:lambda} and Theorem \ref{fst} show that  we can use a first-order term to represent a normalized proof evidence.
%and thus justify this method.
%thus they provide a theoetical foundation for the method of realizability transformation.

\begin{definition}[Representing First-Order Proof Evidence]
\label{fst:rep}
  Let $\phi$ be a mapping from proof evidence variables to first-order terms. We define 
a representation function $\interp{\cdot}_\phi$ from first-order normal proof evidence to first-order terms. 
  
\begin{itemize}
\item $\interp{a}_\phi = \phi(a)$.
  % \item $\interp{\lambda a_1 ... a_n . p}_{\phi} =
  %   \interp{p}_{\phi[\alpha_1/a_1,..., \alpha_n/a_n]}$.

\item $\interp{\kappa \ p_1 ...p_n}_\phi =
  K_{\kappa}(\interp{p_1}_\phi,..., \interp{p_n}_\phi)$, where
  $K_\kappa$ is a function symbol.
\end{itemize}

\end{definition}

%\begin{definition}
  Let $A \equiv P(t_1,..., t_n)$ be an atomic formula and $t'$ be a term such that $(\bigcup_i \mathrm{FV}(t_i)) \cap \mathrm{FV}(t') = \emptyset$, we write
  $A[t']$ to 
  abbreviate a new atomic formula $P(t_1,..., t_n, t')$.
%\end{definition}

\begin{definition}[Realizability Transformation]
\label{real}
  We define a transformation $F$ on Horn formula and its normalized proof evidence: 
  \begin{itemize}
  \item $F(\kappa : \forall \underline{x} . A_1, ..., A_m \Rightarrow B) = \kappa : \forall \underline{x} . \forall \underline{y}. A_1[y_1], ..., A_m[y_m] \Rightarrow B[K_\kappa(y_1,...,y_m)]$, where $y_1,..., y_m$ are all fresh and distinct.
  \item $F(\lambda \underline{a} . n : [\forall \underline{x}] . A_1, ..., A_m \Rightarrow B) = \lambda \underline{a} . n : [\forall \underline{x}.\forall \underline{y}]. A_1[y_1], ..., A_m[y_m] \Rightarrow B[\interp{n}_{[\underline{y}/\underline{a}]}]$, where $y_1,..., y_m$ are all fresh and distinct.
  \end{itemize}
     
\end{definition}

% is crutial to realizability transformation, since normalized proof can be represented by a first-order term. 

% The following theorems show the realizability transformation has three highly desirable properties, namely,
% the transformation does not change the proof theoretic meanings for LP-Unif, it yields a finite term-matching reduction for LP-TM, and it allows us to compute the proof evidence for the goal formula for LP-Unif. 
%% The following theorem shows the realizability transformation yields a finite term-matching reduction for LP-TM. 
The realizability transformation systematically associates a proof to each predicate,
so that the proof can be recorded alongside with reductions. 
Let $F(\Phi)$ mean applying the realizability transformation to every axiom in $\Phi$.
\begin{example}
 \label{ex:conn:real0}
The following logic program $F(\Phi)$ is the result of applying realizability transformation on
the program $\Phi$ in Example \ref{ex:conn}.

  \begin{center}
  $  \kappa_1 : \forall x . \forall y . \forall u_1. \forall u_2 . {Connect}(x, y, u_1), {Connect}(y, z, u_2) \Rightarrow {Connect}(x, z, K_{\kappa_1}(u_1, u_2))$

  $\kappa_2 : \ \Rightarrow {Connect}({Node_1}, {Node_2}, K_{\kappa_2})$
  
  $\kappa_3 : \ \Rightarrow {Connect}({Node_2}, {Node_3}, K_{\kappa_3})$
    \end{center}

\noindent Before the realizability transformation, we have the following judgement in \textbf{H}:

\begin{center}
  $\Phi \vdash \lambda b. (\kappa_1\ b)\ \kappa_2 :
  {Connect}({Node_2}, z) \Rightarrow
  {Connect}({Node_1}, z)$
\end{center}

\noindent We can apply the transformation, we get: 

\begin{center}
  $F(\Phi) \vdash \lambda b. (\kappa_1\ b)\ \kappa_2 :
  {Connect}({Node_2}, z, u_1) \Rightarrow
  {Connect}({Node_1}, z, \interp{(\kappa_1\ b)\ \kappa_2}_{[u_1/b]})$
\end{center}

\noindent which is the same as

\begin{center}
  $F(\Phi) \vdash \lambda b. (\kappa_1\ b)\ \kappa_2 :
  {Connect}({Node_2}, z, u_1) \Rightarrow
  {Connect}({Node_1}, z, K_{\kappa_1}( u_1, K_{\kappa_2}))$
\end{center}

%% \noindent Observe that the transformed formula:

%% \noindent $\mathrm{Connect}(\mathrm{node_2}, z, u_1) \Rightarrow \mathrm{Connect}(\mathrm{node_1}, z, f_{\kappa_1}( u_1, c_{\kappa_2}))$ is provable by $\lambda b. (\kappa_1\ b)\ \kappa_2$ using the transformed program.
\end{example}

 We write $(F(\Phi), \leadsto)$, to
mean given axioms $F(\Phi)$, use LP-Unif to reduce a given query. Note that for a query $A$ in $(\Phi, \leadsto)$, it becomes a query $A[t]$ for some $t$ such that $\mathrm{FV}(A) \cap \mathrm{FV}(t) = \emptyset$ in $(F(\Phi), \leadsto)$.

 The following theorem shows that realizability transformation does not change the type-theoretic meaning of a program. This is important because it means we can apply different resolution strategies to resolve the query on the transformed program without worrying about the change of meaning. Later we will see that the behavior of LP-Struct
is different for the original program and the transformed program.

\begin{theorem}\label{th6}
\label{realI}
If $\Phi \vdash e: [\forall \underline{x}] . \underline{A}\Rightarrow B$ in $\mathbf{C}$ and $e$
normalized to $n$, then $F(\Phi) \vdash F(n : [\forall \underline{x}] . \underline{A}\Rightarrow B)$ in $\mathbf{H}$. 
\end{theorem}
\begin{proof}
  By induction on the derivation of $\Phi \vdash e : [\forall \underline{x}] . \underline{A}\Rightarrow B$.
  \begin{itemize}
  \item Base Case. 

   \

    \begin{tabular}{l}
     \infer{\Phi \vdash \kappa : \forall \underline{x} . \underline{A}\Rightarrow B}{(\kappa : \forall \underline{x} . \underline{A}\Rightarrow B) \in \Phi}      

    \end{tabular}
   
\

 In this case, we know that $F(\kappa : \forall \underline{x} . \underline{A}\Rightarrow B) = \kappa : \forall \underline{x}. \forall \underline{y} . A_1[y_1], ..., A_n[y_n] \Rightarrow B[f_\kappa (y_1,..., y_n)] \in F(\Phi)$. 

\item Step Case. 

\

\begin{tabular}{l}
\infer{\Phi \vdash \lambda \underline{a} . \lambda \underline{b} . (e_2\ \underline{b})\ (e_1\ \underline{a}) : \underline{A}, \underline{B} \Rightarrow C}{\Phi \vdash e_1 : \underline{A} \Rightarrow D & \Phi \vdash e_2 : \underline{B}, D \Rightarrow C}
  
\end{tabular}

\

By Lemma \ref{fst:lambda}, we know that the normal form of $e_1$ is $\kappa_1$ or $\lambda \underline{a}. n_1$, and the normal form of $e_1$ is $\kappa_2$ or $\lambda \underline{b} d. n_2$, with $n_1, n_2$ are first-order. 
\begin{itemize}
\item $e_1 \equiv \kappa_1, e_2 \equiv \kappa_2$. By IH, we know that $\Phi \vdash \kappa_1 : A_1[y_1],..., A_n[y_n] \Rightarrow D[f_{\kappa_1}(y_1,..., y_n)] $ and 
$\Phi \vdash \kappa_2 : B_1[z_1],..., B_m[z_m], D[y] \Rightarrow C[f_{\kappa_2}(z_1,..., z_m, y)]$. So by \textsc{Gen} and \textsc{Inst}, we have

\noindent $\Phi \vdash \kappa_2 : B_1[z_1],..., B_m[z_m], D[f_{\kappa_1}(y_1,..., y_n)] \Rightarrow C[f_{\kappa_2}(\underline{z}, f_{\kappa_1}(\underline{y}))]$. 

\noindent Then by the \textsc{Cut} rule, we have 

\noindent $\Phi \vdash \lambda \underline{a} . \lambda \underline{b} . \kappa_2\ \underline{b}\ (\kappa_1\ \underline{a}) : A_1[y_1],..., A_n[y_n],  B_1[z_1],..., B_m[z_m] \Rightarrow C[f_{\kappa_2}(\underline{z}, f_{\kappa_1}(\underline{y}))]$. 

\noindent We can see that $\interp{\kappa_2\ \underline{b}\ (\kappa_1\ \underline{a})}_{[\underline{y}/\underline{a}, \underline{z}/\underline{b}]} = f_{\kappa_2}(\underline{z}, f_{\kappa_1}(\underline{y}))$.

\item $e_1 \equiv \lambda \underline{a}. n_1, e_2 \equiv \lambda \underline{b}.d. n_2$. By IH, we know that $\Phi \vdash \lambda \underline{a}. n_1 : A_1[y_1],..., A_1[y_1] \Rightarrow D[\interp{n_1}_{[\underline{y}/\underline{a}]}]$ and 
$\Phi \vdash \lambda \underline{b}. d. n_2 : B_1[z_1],..., B_m[z_m], D[y] \Rightarrow C[\interp{n_2}_{[\underline{z}/ \underline{b}, y/d]}]$. So by \textsc{Gen} and \textsc{Inst}, we have

\noindent $\Phi \vdash \lambda \underline{b}. d. n_2 : B_1[z_1],..., B_m[z_m], D[\interp{n_1}_{[\underline{y}/\underline{a}]}] \Rightarrow C[\interp{n_2}_{[\underline{z}/\underline{b}, \interp{n_1}_{[\underline{y}/\underline{a}]}/d]}]$. 

\noindent Then by the \textsc{Cut} rule and beta reduction, we have

\noindent $\Phi \vdash \lambda \underline{a} . \lambda \underline{b} . ([n_1/d]n_2) : A_1[y_1],..., A_1[y_1],  B_1[z_1],..., B_m[z_m] \Rightarrow C[\interp{n_2}_{[\underline{z}/\underline{b}, \interp{n_1}_{[\underline{y}/\underline{a}]}/d]}]$ in $\mathbf{H}$. We 
know that $\interp{[n_1/d]n_2}_{[\underline{y}/\underline{a}, \underline{z}/\underline{b}]} = \interp{n_2}_{[\underline{z}/\underline{b}, \interp{n_1}_{[\underline{y}/\underline{a}]}/d]}$. 
\item The other cases are handle similarly. 
\end{itemize}

\item Step Case.

\

  \begin{tabular}{l}
\infer{\Phi \vdash \lambda \underline{a} . n : [\underline{t}/\underline{x}]\underline{A} \Rightarrow [\underline{t}/\underline{x}]B }{\Phi \vdash \lambda \underline{a} . n : \forall \underline{x} . \underline{A} \Rightarrow B}
\end{tabular}

\

By IH, we know that $\Phi \vdash \lambda \underline{a} . n : \forall \underline{x}. \forall \underline{y} . A_1[y_1],..., A_n[y_n] \Rightarrow B[\interp{n}_{[\underline{y}/\underline{a}]}]$. By \textsc{Inst} rule,  we have 
$\Phi \vdash \lambda \underline{a} . n : [\underline{t}/\underline{x}] A_1[y_1],..., [\underline{t}/\underline{x}] A_n[y_n] \Rightarrow [\underline{t}/\underline{x}] B[\interp{n}_{[\underline{y}/\underline{a}]}]$

\item Step Case.

\

  \begin{tabular}{l}
\infer{\Phi \vdash e: \forall \underline{x} . F}{\Phi \vdash e : F}
\end{tabular}

\

This case is straightforwardly by IH. 
  \end{itemize}
\end{proof}

 The following lemma and a theorem show that the extra argument can be used to record the term representation of the corresponding proof.
\begin{lemma}
\label{realII}
  If $F(\Phi) \vdash \{A_1[y_1],..., A_n[y_n]\} \leadsto^*_{\gamma} \emptyset$, and $y_1,..., y_n$ are fresh, then $F(\Phi) \vdash e_i : \forall \underline{x} . \Rightarrow \gamma A_i[\gamma y_i]$ in $\mathbf{H}$ with $\interp{e_i}_{\emptyset} = \gamma y_i $ for all $i$.  
\end{lemma}
\begin{proof}
  By induction on the length of the reduction $F(\Phi) \vdash \{A_1[y_1],..., A_n[y_n]\} \leadsto^*_{\gamma} \emptyset$. 
  \begin{itemize}
  \item Base Case. Suppose the length is one, namely, $F(\Phi) \vdash \{A[y]\} \leadsto_{\kappa, \gamma_1} \emptyset$. Thus there exists $(\kappa : \forall \underline{x} .  \Rightarrow C[f_\kappa]) \in F(\Phi)$(here $f_\kappa$ is a constant), such that $C[f_\kappa] \sim_{\gamma_1} A[y]$.  Thus $ \gamma_1 (C[f_\kappa]) \equiv \gamma_1 A[\gamma_1 y]$. So $\gamma_1 y \equiv f_\kappa$
and $\gamma_1 C \equiv \gamma_1 A$. We have $F(\Phi) \vdash \kappa :\ \Rightarrow  \gamma_1 C[f_\kappa]$ by the \textsc{Inst} rule, thus $F(\Phi) \vdash \kappa :\ \Rightarrow \gamma_1 A[\gamma_1 y]$, hence $F(\Phi) \vdash \kappa : \forall \underline{x} . \Rightarrow \gamma_1 A[\gamma_1 y]$ by the \textsc{Gen} rule and $\interp{\kappa}_{\emptyset} = f_{\kappa}$.

  \item Step Case. Suppose 

\noindent $F(\Phi) \vdash \{A_1[y_1], ..., A_i[y_i],..., A_n[y_n]\} \leadsto_{\kappa, \gamma_1} \{\gamma_1 A_1[y_1],...,  \gamma_1 B_1[z_1],...,  \gamma_1 B_m[z_m],..., \gamma_1 A_n[y_n]\} \leadsto^*_{\gamma} \emptyset$,
 where $\kappa : \forall \underline{x} . \forall \underline{z} . B_1[z_1],..., B_m[z_m] \Rightarrow C[f_\kappa(z_1,..., z_m)] \in F(\Phi)$, and $C[f_\kappa(z_1,..., z_m)] \sim_{ \gamma_1} A_i[y_i]$. So we know $ \gamma_1 C[f_\kappa(z_1,..., z_m)] \equiv \gamma_1 A_i[\gamma_1 y_i]$,  $\gamma_1 y_i \equiv f_\kappa(z_1,..., z_m),  \gamma_1 C \equiv \gamma_1 A_i$ and 

\noindent $\mathrm{dom}(\gamma_1) \cap \{z_1,..., z_m, y_1,..,y_{i-1}, y_{i+1}, y_n\} = \emptyset$. By IH, we know that  $F(\Phi) \vdash e_1 : \forall \underline{x}. \Rightarrow \gamma \gamma_1 A_1[\gamma y_1],..., $

\noindent $F(\Phi) \vdash p_1 : \forall \underline{x}. \Rightarrow \gamma   \gamma_1 B_1[\gamma z_1],..., F(\Phi) \vdash p_m : \forall \underline{x}. \Rightarrow \gamma   \gamma_1 B_m[\gamma z_m],..., F(\Phi) \vdash e_n : \forall \underline{x} . \Rightarrow \gamma \gamma_1 A_n[\gamma y_n]$ and $\interp{e_1}_{\emptyset} = \gamma y_1, ..., \interp{p_1}_\emptyset = \gamma z_1, ..., \interp{p_m}_\emptyset = \gamma z_m, ..., \interp{e_n}_{\emptyset} = \gamma y_n $ . We can construct a proof $e_i = \kappa \ p_1\ ... p_m$ with $e_i : \forall \underline{x} . \Rightarrow \gamma \gamma_1 A_i[\gamma \gamma_1 y_i]$, by first apply the \textsc{Inst} to instantiate the quantifiers of $\kappa$, then applying the \textsc{Cut} rule $m$ times. Moreover, we have $\interp{\kappa \ p_1\ ... p_m}_{\emptyset} = f_\kappa(\interp{p_1}_\emptyset,...,\interp{p_m}_\emptyset) = \gamma (f_\kappa (z_1,..., z_m)) = \gamma \gamma_1 y_i$. 
     \end{itemize}

\end{proof}

\begin{theorem}\label{th7}
\label{record}
 Suppose $F(\Phi) \vdash \{A[y]\} \leadsto^*_{\gamma} \emptyset$. We have $F(\Phi) \vdash p : \forall \underline{x} . \Rightarrow \gamma A[\gamma y]$ in $\mathbf{H}$, where $p$ is in normal form and $\interp{p}_{\emptyset} = \gamma y$. 
\end{theorem} 

Now we are able to show that realizability transformation will not change the unification reduction behaviour.
\begin{lemma}
If $\Phi \vdash \{A_1,..., A_n\} \leadsto^* \emptyset$, then $F(\Phi) \vdash \{A_1[y_1],..., A_n[y_n]\} \leadsto^* \emptyset$ with $y_i$ fresh for all $i$. 
\end{lemma}
\begin{proof}
  By induction on the length of $\Phi \vdash \{A_1,..., A_n\} \leadsto^* \emptyset$.
  \begin{itemize}
     \item Base Case. Suppose the length is one, namely, $\Phi \vdash \{A\} \leadsto_{\kappa, \gamma_1} \emptyset$. There exists $(\kappa : \forall \underline{x} .\  \Rightarrow C) \in \Phi$ such that $C \sim_{\gamma_1} A$.  Thus $\kappa : \forall \underline{x}. \ \Rightarrow C[f_\kappa] \in F(\Phi)$ and $(C[f_\kappa]) \sim_{\gamma_1[f_\kappa/y]} A[y]$. So $F(\Phi) \vdash \{A[y] \} \leadsto_{\kappa,\gamma_1[f_\kappa/y]} \emptyset$.
     \item Step Case. Suppose $\Phi \vdash \{A_1, ..., A_i,..., A_n\} \leadsto_{\kappa, \gamma_1} \{\gamma_1 A_1,...,  \gamma_1 B_1,...,  \gamma_1 B_m,..., \gamma_1 A_n\} \leadsto^*_{\gamma} \emptyset$, where $\kappa : \forall \underline{x}. B_1,..., B_m \Rightarrow C \in \Phi$ and $C \sim_{ \gamma_1} A_i$. So we know that $\kappa : \forall \underline{x}. \forall \underline{z}. B_1[z_1],..., B_m[z_m] \Rightarrow C[f_\kappa(\underline{z})] \in F(\Phi)$ and $C[f_\kappa(\underline{z})] \sim_{\gamma_1[f_\kappa(\underline{z})/y_i]} A_i[y_i]$. Thus we have the following reduction:

       \begin{center}
         $F(\Phi) \vdash \{A_1[y_1], ..., A_i[y_i],..., A_n[y_n]\}
         \leadsto_{\kappa, \gamma_1[f_\kappa(\underline{z})/y_i]}
         \{\gamma_1[f_\kappa(\underline{z})/y_i] A_1[y_1],...,
         \gamma_1[f_\kappa(\underline{z})/y_i] B_1[z_1],...,
         \gamma_1[f_\kappa(\underline{z})/y_i] B_m[z_m],...,
         \gamma_1[f_\kappa(\underline{z})/y_i] A_n[y_n]\} \equiv
         \{\gamma_1 A_1[y_1],..., \gamma_1 B_1[z_1],..., \gamma_1
         B_m[z_m],..., \gamma_1 A_n[y_n]\}$
       \end{center}
\noindent By IH, $F(\Phi) \vdash \{\gamma_1 A_1[y_1],...,  \gamma_1 B_1[z_1],...,  \gamma_1 B_m[z_m],..., \gamma_1 A_n[y_n]\} \leadsto^* \emptyset$.

  \end{itemize}
\end{proof}

\begin{lemma}
If $F(\Phi) \vdash \{A_1[y_1],..., A_n[y_n]\} \leadsto^* \emptyset$ with $y_i$ fresh for all $i$, then $\Phi \vdash \{A_1,..., A_n\} \leadsto^* \emptyset$. 
\end{lemma}
\begin{proof}
  By induction on the length of $F(\Phi) \vdash \{A_1[y_1],..., A_n[y_n]\} \leadsto^* \emptyset$.
  \begin{itemize}
     \item Base Case. Suppose the length is one, namely, $F(\Phi) \vdash \{A[y]\} \leadsto_{\kappa, \gamma_1} \emptyset$. 

We know that $(\kappa : \forall \underline{x} .  \Rightarrow C[f_\kappa]) \in F(\Phi)$ with $C[f_\kappa] \sim_{\gamma_1} A[y]$.  Thus $  C \sim_{\gamma_1-[f_\kappa/y]} A$. So $\Phi \vdash \{A\}\leadsto \emptyset$.
     \item Step Case. Suppose we have the following reduction: 

 $F(\Phi) \vdash \{A_1[y_1], ..., A_i[y_i],..., A_n[y_n]\} \leadsto_{\kappa, \gamma_1} \{\gamma_1 A_1[y_1],...,  \gamma_1 B_1[z_1],...,  \gamma_1 B_m[z_m],..., \gamma_1 A_n[y_n]\} \leadsto^*_{\gamma} \emptyset$

\noindent  Note that $\kappa : \forall \underline{x} . \forall \underline{z} . B_1[z_m],..., B_m[z_m] \Rightarrow C[f_\kappa(z_1,..., z_m)] \in F(\Phi)$ and $C[f_\kappa(z_1,..., z_m)] \sim_{ \gamma_1} A_i[y_i]$. So we know $C \sim_{\gamma_1-[f_\kappa(\underline{z})/y_i]} A_i$. Let $\gamma = \gamma_1-[f_\kappa(\underline{z})/y_i]$. We have

\begin{center}
  $\Phi \vdash \{A_1,..., A_i,..., A_n\} \leadsto_{\kappa, \gamma} \{\gamma A_1,...,
  \gamma B_1,..., \gamma B_m, ..., \gamma A_n\} \equiv \{\gamma_1
  A_1,..., \gamma_1 B_1,..., \gamma_1 B_m, ..., \gamma_1 A_n\}$
\end{center}
 By IH, we know $\Phi \vdash \{\gamma_1 A_1,..., \gamma_1 B_1,..., \gamma_1 B_m, ..., \gamma_1 A_n\} \leadsto^* \emptyset$.
  \end{itemize}
\end{proof}

\begin{theorem}\label{th8}
\label{preservation}
  $\Phi \vdash \{A\} \leadsto^* \emptyset$ iff $F(\Phi) \vdash \{A[y]\} \leadsto^* \emptyset$. 
\end{theorem}

\begin{example}
 \label{ex:conn:real}
Consider the logic program in Example \ref{ex:conn:real0}. Realizability transformation does not change the behaviour of LP-Unif, we still have the 
  following successful unification reduction path for query ${Connect}(x, y, u)$:
   
  \begin{center}

$F(\Phi) \vdash \{{Connect}(x, y, u)\}\leadsto_{\kappa_1, [x/x_1, y/z_1, K_{\kappa_1}(u_3, u_4)/u]} \{{Connect}(x, y_1, u_3), {Connect}(y_1, y, u_4)\}$

$\leadsto_{\kappa_2, [K_{\kappa_2}/u_3,{Node_1}/x, {Node_2}/y_1, {Node_1}/x_1, b/z_1, K_{\kappa_1}(K_{\kappa_2}, u_4)/u]} $

$\{{Connect}({Node_2}, y, u_4)\}$

$\leadsto_{\kappa_3, [K_{\kappa_3}/u_4, K_{\kappa_2}/u_3, {Node_3}/y, {Node_1}/x, {Node_2}/y_1,{Node_1}/x_1, {Node_3}/z_1, K_{\kappa_1}(K_{\kappa_2}, K_{\kappa_3})/u]} \emptyset $

\end{center}
 
\end{example}

There are logic programs that are overlapping and LP-TM nonterminating (as e.g. the program of Example \ref{ex:conn}),  we would still like
 to obtain a meaningful execution behaviour for LP-Struct, especially if LP-Unif aready allows successful derivations for the programs.
Luckily, we can apply realizability transformation to such programs and apply LP-Struct reduction. % the overlapping and LP-TM nonterminating program.
 %For these program,
 %we can first apply the
%realizability transformation to the program, and we can always do this since it does not change the 
%meaning of the program, and then we can use LP-Struct on the transformed program.  

\begin{proposition}
  \label{trans:equiv}
 For any program $\Phi$, $F(\Phi)$ is LP-TM terminating and non-overlapping.
\end{proposition} 
\begin{proof}
First, we need to show $\to$-reduction is strongly normalizing in $(F(\Phi), \to)$. By Definition \ref{real}, we can establish a decreasing measurement(from right to left, using the strict subterm relation) for each rule in $F(\Phi)$, since the last argument in the head of each rule is strictly larger than the ones in the body. Then, non-overlapping property is due to the fact that all the heads of the rules in $F(\Phi)$ will be \textit{guarded} by the unique function symbol in Definition \ref{real}.
\end{proof}

%The above results allows us to prove formally the equivalence of LP-Struct and LP-Unif. 

\begin{corollary}[Equivalence of LP-Unif and LP-Struct]
  \label{equiv}
 $F(\Phi) \vdash \{A_1,..., A_n\} (\to^\mu \cdot \hookrightarrow^1)^* \{B_1,..., B_m\}$ iff $F(\Phi) \vdash \{A_1,..., A_n\} \leadsto^* \{B_1,..., B_m\}$.  
\end{corollary}
\begin{proof}
  By Theorem \ref{prod-non-overlap} and Proposition \ref{trans:equiv}. 
\end{proof}

Using the above corollary and soundness and completeness of LP-Unif, we deduce as a corollary that LP-Struct is sound and complete relative to system $\mathbf{H}$ for the transformed logic program.  

\begin{example}
\label{ex:conn:real1}
For the program in Example \ref{ex:conn:real0}, the query $\mathrm{Connect}(x, y, u)$ can be reduced by LP-Struct successfully:   
 
  \begin{center}

$F(\Phi) \vdash \{{Connect}(x, y, u)\} \hookrightarrow_{\kappa_1, [x/x_1, y/z_1, K_{\kappa_1}(u_3, u_4)/u]} \{{Connect}(x, y,K_{\kappa_1}(u_3, u_4))\} \to_{\kappa_1} \{{Connect}(x, y_1, u_3), {Connect}(y_1, y, u_4)\}$

$\hookrightarrow_{\kappa_2, [K_{\kappa_2}/u_3,{Node_1}/x, {Node_2}/y_1, {Node_1}/x_1, b/z_1, K_{\kappa_1}(K_{\kappa_2}, u_4)/u]} \{{Connect}({Node_1},{Node_2}, K_{\kappa_2}), {Connect}({Node_2}, y, u_4)\} \to_{\kappa_2} \{{Connect}({Node_2}, y, u_4)\}$

$\hookrightarrow_{\kappa_3, [K_{\kappa_3}/u_4, K_{\kappa_2}/u_3, {Node_3}/y, {Node_1}/x, {Node_2}/y_1,{Node_1}/x_1, {Node_3}/z_1, K_{\kappa_1}(K_{\kappa_2},K_{\kappa_3})/u]}  \{{Connect}({Node_2}, {Node_3}, K_{\kappa_3})\}  \to_{\kappa_3} \emptyset $
\end{center}
   
\noindent Note that the answer for $u$ is $K_{\kappa_1}(K_{\kappa_2},K_{\kappa_3})$, which is the first-order term representation of the proof of $ \ \Rightarrow {Connect}({Node}_1, {Node}_3)$. 
\end{example}

Realizability transformation uses the extra argument as decreasing measurement
in the program to achieve termination of $\to$-reduction.  At the same time this extra argument makes the program non-overlapping.  %We want to point out that
Realizability 
transformation does not modify the proof-theoretic meaning and the execution behaviour of LP-Unif. 
The next example shows that not every transformation technique for obtaining structurally decreasing LP-TM reductions has such properties:

\begin{example}
  \label{overlap}
  
  Consider the following program:

\begin{center}
\noindent $\kappa_1 : \ \Rightarrow P({Int})$  

\noindent $\kappa_2 : \forall x . P(x), P({List}(x)) \Rightarrow P({List}(x))$
\end{center} 

\noindent It is a folklore method to add a structurally decreasing argument as a  measurement to ensure finiteness of $\to^\mu$.

\begin{center}

\noindent $\kappa_1 : \ \Rightarrow P({Int}, 0)$  

\noindent $\kappa_2 : \forall x . \forall y . P(x, y), P({List}(x), y) \Rightarrow P({List}(x), {S}(y))$
\end{center} 

\noindent We denote the above program as $\Phi'$. Indeed with the measurement we add, the term-matching reduction in $\Phi'$ will be finite. But the reduction for query $P(\mathrm{List}(\mathrm{Int}), z)$ using LP-Unif reduction will fail:

\begin{center}
$\Phi' \vdash \{P({List}({Int}), z) \}\leadsto_{\kappa_2, [ {Int}/x,{S}(y_1)/z]} \{P({Int}, y_1), P({List}({Int}), y_1)\}\leadsto_{\kappa_2, [ 0/y_1, {Int}/x,{S}(0)/z]} \{P({List}({Int}), 0)\} \not \leadsto$ 
\end{center}

\noindent However, the query $P({List}({Int}))$ on the original program using unification reduction will diverge. Divergence and failure are operationally different. Thus
adding arbitrary measurement may modify the execution behaviour of a program (and hence the meaning of the program). In contrast, 
by Theorems~\ref{th6}-\ref{th8},
realizability transformation does not modify the execution behaviour of LP-Unif reduction.
\end{example}

\begin{example}
Consider the following non-LP-TM terminating and non-overlapping program and its version after the realizability transformation:

\begin{center}
\emph{Original program:}  $\kappa : \forall x . P(x) \Rightarrow P(x)$\\
%\end{center}
%\begin{center}
\emph{After transformation:}  $\kappa : \forall x . \forall u. P(x, u) \Rightarrow P(x, K_\kappa(u))$
\end{center}

\noindent Both LP-Struct and LP-Unif will diverge for the queries $P(x), P(x, y)$ in both original and transformed versions. LP-Struct reduction diverges for different reasons in the two cases, one is due to divergence of $\to$-reduction: $\Phi \vdash \{P(x)\} \to \{P(x)\} \to \{P(x)\} \to ...$. The another is due to $\hookrightarrow$-reduction: $\Phi \vdash \{P(x, y)\} \hookrightarrow \{P(x, f_k(u))\} \to \{P(x, u)\} \hookrightarrow \{P(x, K_k(u'))\} \to \{P(x, u')\} \hookrightarrow ...$. Note that a single step of LP-Unif reduction for the original program corresponds to infinite steps of term-matching reduction in LP-Struct. For the transformed version, a single step of LP-Unif reduction corresponds to finite steps of LP-Struct reduction.
\end{example}

\section{Related Work}
\label{rw}
\textit{Proof Search, Logic Programming and Type Theory.} To the best of our knowledge, studying logic programming proof-theoretically
dates back to Girard's suggestion to use the cut rule to model resolution for Horn formulas \cite[Chapter 13.4]{Girard:1989}. Miller et. al. \cite{Miller:1989} use cut-free sequent calculus to
represent a proof for a query. More specifically, given a query $Q$ and a
logic program $\mathcal{P}$, %that consists of a Horn formula $A_1 \wedge ... \wedge A_n \Rightarrow B$, if the query
$Q$ has a refutation iff there is a derivation in cut-free sequent calculus for $\mathcal{P} \vdash Q$.
Using sequent calculus as a proof theoretic framework gives the flexibility to incorporate different kinds of formulas, e.g. classical formulas and linear formulas into this framework. 

 Interactive theorem prover 
Twelf \cite{pfenning1999system} pioneered  implementation of proof search on top of a depedently typed system 
called LF \cite{harper2007mechanizing}. Similar to Twelf, we believe that type systems serve
as a suitable foundation for logic programming. Comparing to Twelf,
we  specify and analyze different resolution strategies (other than SLD-resolution) and
study their intrinsic relations. 
%% Proof evidence in its connection with type system is also studied in Mark Jones's
%% thesis~\cite[Chapter 4.2]{jones2003qualified} in the context of type class resolution.
We also pay more attention to various kinds of productivity compared to Twelf. 

\textit{Structural Resolution.} Structural resolution is a result of joint research efforts by Komendantskaya et. al. (\cite{JKK15}, \cite{KomendantskayaP11}, \cite{komendantskaya2014}). 
The goal of the analysis of structural resolution is to support \textit{sound} coinductive reasoning in logic programming.
For example, given the query $Take(S(S(S(Z))),y, z ), Fib(A,B,y)$ in Example \ref{fib}, one may want not only to obtain a substitution for $z$,
but also a guarantee that the queries to $Fib$ are nonterminating and, moreover, that derivations for $Fib$ will not fail if continued to infinity.
To support this, \textit{productivity analysis} has been developed~\cite{KJ15,KJS16} as a compile time technique to detect observational productivity of logic programs.
%of a given logic program. 
%    JohannKK15

%% \fp{Katya, could you give a quick introduction on current state of the art on structural resolution and its challenge? }

\textit{Coinductive Logic Programming.} Gupta et al. \cite{Gupta07}'s \textit{coinductive logic programming} (CoLP) extends  SLD-resolution with a method to use
atomic \textit{coinductive hypotheses}.  That is,
during the execution, if the current queries $\{ C_1, ..., C_i, ..., C_n\}$ contain a query $C_i$ that unifies via $\gamma$ with a $C'$ in the earlier execution, %(with the substitution $\gamma$ and occur
%check removed),
then the next step of resolution  will be given by $\{\gamma C_1, ..., \gamma C_{i-1}, \gamma C_{i+1}, ..., \gamma C_n\}$. The coinductvie hyposesis mechanism in CoLP can be viewed as a form of
loop detection. However, CoLP cannot detect hypotheses for more complex patterns of coinduction that produce coinductive subgoals that fail to unify. As discussed in introduction, it is not a suitable tool
to analyze the productivity of infinite data structures in logic programming.
%is that not every nonterminating resolution 
%can exibit such loop, in such case it may be difficult to understand the nontermination using CoLP. 

\textit{Proof Relevant Corecursive Resolution.} In our previous work \cite{flops16}, we extended system $\mathbf{H}$
  with fixpoint operator to allow constructing  \textit{corecursive proof evidence} (given by proof terms containing fixpoint operator) for certain
nonterminating LP-TM reductions. The type system that we use to justify the corecursive proof evidence is an extension of Howard's system \textbf{H} with the fixpoint typing 
rule. There, the main challenge was to heuristically and automatically construct corecursive evidence for a given query. % and thus to support automation of such proofs.
%% In general, the problem of constructing a recursive proof term for nonterminating queries is equivalent to generation of recursive schemes and is
%% therefore undecidable. 

\textit{Logic Programming by Term Matching.}
LP-TM reduction may seem to be a rare kind of resolution, but it underlies many 
applications. 
The process of simplifying type class
constraints is formally described as the notion of \textit{context reduction}
by Peyton Jones et. al.~\cite{Jones97}. The context reduction process uses exactly the LP-TM reduction
that we described in this paper. The logic-based multi-paradigm programming language PiCAT \cite{zhou2013user, zhou2015constraint} makes extensive use of term-matching with explicit unification. For example, the Fibonacci
sequence in PiCAT is defined as follows: 

\begin{verbatim}
fib(0,F) => F=1.
fib(1,F) => F=1.
fib(N,F),N>1 => fib(N-1,F1),fib(N-2,F2),F=F1+F2.
\end{verbatim}
Through the functionalisation process, the existing termination detection techniques from term rewriting~\cite{bezem2003term} can be directly applied to LP-TM.
Thus we think our work in Section \ref{s:func} builds a useful link between LP-TM and term rewriting.
% Type class
% evidence in its connection with type system is studied in Mark Jones's
% thesis~\cite[Chapter 4.2]{jones2003qualified}. Instance declarations can also be
% interpreted as single head \textit{simplification} rules in Constraints Handling Rules (CHR)
% \cite{SulzmannDJS07}, which implies that instance declarations can be modeled as Horn formulas naturally.  

% To our knowledge, the tradition of studying logic programming proof-theoretically
% dates back to Girard's suggestion that the cut rule can model resolution for Horn formulas \cite[Chapter 13.4]{girard1989proofs}. Alternatively,   
%  Miller et. al. \cite{miller1991uniform} model Horn formulas using cut-free sequent calculus.   
% Context reduction, instance declaration and their connection to proof relevant
% resolution are also discussed under the name of \textit{LP-TM} (logic programming with
% term-matching) in Fu and Komendantskaya~\cite[Section 4.1]{FuK15}.        
% \textit{Termination of Term Rewriting System} Since LP-TM without existential variable can be transformed
% to term rewriting system via functionalisation, all the termination techniques developed in term rewriting
% literature can be directly applied to show the termination of LP-TM. Other than the dependency pair method
% \cite{Arts2000}, one can also uses lexicographic ordering,  

\section{Conclusions}\label{concl}

We have shown that Howard's system $\mathbf{H}$ is a suitable foundation for logic programming. 
We have proven soundness and completeness of LP-Unif with
respect to the type system $\mathbf{H}$.
We have developed a partial LP-Unif resolution strategy based on labels
    to control LP-Unif reductions and achieve a form of lazy computation. Based on partial
    LP-Unif, we have also defined a new  notion of local productivity. 

We have formally defined structural resolution as LP-Struct, and showed that it in fact combines % which is a combination
   term-matching resolution with unification. We have shown that
  LP-Struct is operationally equivalent to LP-Unif if the program is LP-TM terminating and non-overlapping. Realizability transformation was suggested as an efficient method to
   render all logic programs LP-TM terminating and non-overlapping. 
   We have shown that the realizability transformation preserves the meaning of the logic program relative to $\mathbf{H}$.
   The equivalence of LP-Struct and LP-Unif has been shown, for the transformed program. As a 
   result, we obtained the soundness and completeness of LP-Struct with respect to \textbf{H} as a corollary.

   We have paid a special attention to a study of LP-TM resolution.
   We have defined a process called functionalisation that transforms a logic program without existential variables into a term rewriting system.
   We have shown the exact relation of LP-TM and term rewriting system, and gave an example of using dependency pair technique from term rewriting to detect the termination of LP-TM.

    For future work, we would like to provide a method to establish local productivity
    for a given query and study the relation between global productivity and local productivity
    in more detail. We plan to implement partial LP-Unif and explore its implications.

    \subsection*{Acknowledgements} We thank the LOPSTR'15 anonymous referees for their helpful comments and suggestions for the earlier conference versions of this paper. 
    We thank Tom Schrijvers for exposing us to the nonterminating type class resolution problem
    and many helpful discussions. This work was funded by EPSRC grant EP/K031864/1-2. A large part of this research was done at the University of Dundee.  

\bibliographystyle{alpha}
 
% The bibliography should be embedded for final submission.
\bibliography{fac-journal} 

\end{document}